\documentclass[11pt,letterpaper]{article}

\usepackage{amsmath,amsfonts,amsthm,amssymb,stmaryrd,graphicx,relsize,bm}  
\theoremstyle{plain}

\numberwithin{equation}{section}

\newtheorem{thm}{Theorem}[section]
\newtheorem{lem}[thm]{Lemma}
\newtheorem{cor}[thm]{Corollary}

{\end{flushleft}}
\newenvironment{conj}[1]
{\begin{flushleft}\textbf{Conjecture #1}.\enspace}%
{\end{flushleft}}

\allowdisplaybreaks  

\newcommand{\tbullet}{\mathrel{\raise .4ex\hbox{\tiny$\bullet$}}} 

\newcommand{\rmtr}{\mathrm{tr\,}}
\newcommand{\inset}{\iscript n}
\newcommand{\obset}{\oscript b}
\newcommand{\escript}{\mathcal{E}}
\newcommand{\fscript}{\mathcal{F}}
\newcommand{\hscript}{\mathcal{H}}
\newcommand{\iscript}{\mathcal{I}}
\newcommand{\jscript}{\mathcal{J}}
\newcommand{\lscript}{\mathcal{L}}
\newcommand{\oscript}{\mathcal{O}}
\newcommand{\pscript}{\mathcal{P}}
\newcommand{\rscript}{\mathcal{R}}
\newcommand{\sscript}{\mathcal{S}}
\newcommand{\tscript}{\mathcal{T}}

\newcommand{\hscripthat}{\widehat{\hscript}}
\newcommand{\iscripthat}{\widehat{\iscript}}
\newcommand{\jscripthat}{\widehat{\jscript}}
\newcommand{\lscripthat}{\widehat{\lscript}}
\newcommand{\hhat}{\widehat{H}}
\newcommand{\jhat}{\widehat{J}}
\newcommand{\khat}{\widehat{K}}
\newcommand{\alphatilde}{\widetilde{\alpha}}
\newcommand{\hscriptbar}{\overline{\hscript}}
\newcommand{\iscriptbar}{\overline{\iscript}}
\newcommand{\jscriptbar}{\overline{\jscript}}
\newcommand{\lscriptbar}{\overline{\lscript}}

\newcommand{\ab}[1]{\left|#1\right|}
\newcommand{\brac}[1]{\left\{#1\right\}}
\newcommand{\paren}[1]{\left(#1\right)}
\newcommand{\sqbrac}[1]{\left[#1\right]}
\newcommand{\elbows}[1]{{\left\langle#1\right\rangle}}
\newcommand{\ket}[1]{{\left|#1\right>}}
\newcommand{\bra}[1]{{\left<#1\right|}}

\begin{document}

\title{DUAL INSTRUMENTS AND\\SEQUENTIAL PRODUCTS OF OBSERVABLES}
\author{Stan Gudder\\ Department of Mathematics\\
University of Denver\\ Denver, Colorado 80208\\
sgudder@du.edu}
\date{}
\maketitle

\begin{abstract}
We first show that every operation possesses an unique dual operation and measures an unique effect. If $a$ and $b$ are effects and $J$ is an operation that measures $a$, we define the sequential product of $a$ then $b$ relative to $J$. Properties of the sequential product are derived and are illustrated in terms of L\"uders and Holevo operations. We next extend this work to the theory of instruments and observables. We also define the concept of an instrument (observable) conditioned by another instrument (observable). Identity, state-constant and repeatable instruments are considered. Sequential products of finite observables relative to L\"uders and Holevo instruments are studied.
\end{abstract}

\section{Sequential Products of Effects}  
Let $S$ be a quantum system described by a complex Hilbert space $H$. One of the main points of this article is that the sequential product of two observables for $S$ depends on the instrument $\iscript$ employed to measure the first observable and is independent of the instrument used to measure the second. In this way, the measurement of the first observable influences the measurement of the second but not vice versa. As we shall see, the sequential product is defined in terms of the dual $\iscript ^*$ of $\iscript$.

We denote the set of bounded linear operators on $H$ by $\lscript (H)$ and the set of trace-class operators on $H$ by $\tscript (H)$. For
$A,B\in\lscript (H)$ we write $A\le B$ if $\elbows{\phi ,A\phi}\le\elbows{\phi ,B\phi}$ for all $\phi\in H$. We say that $A\in\lscript (H)$ is \textit{positive} if $A\ge 0$ and $A$ is an \textit{effect} if $0\le A\le I$ where $0,I$ are the zero and identity operators on $H$, respectively
\cite{dl70,hz12,lud51,nc00}. The set of effects on $H$ is denoted by $\escript (H)$. We interpret effects as measurements that have two possible outcomes, true and false. If $a\in\escript (H)$, then its \textit{complement} $a'=I-a$ is true if and only if $a$ is false. If $a,b\in\escript (H)$ and
$a+b\in (H)$ we write $a\perp b$ and interpret $a+b$ as the statistical sum of the measurements $a$ and $b$. Of course, $0\perp a$ for all
$a\in\escript (H)$ and we interpret $0$ as the effect that is always false. Similarly, $1\perp a$ if and only if $a=0$ and $1$ is the effect that is always true. Moreover, $b\perp a$ if and only if $b\le a'$. A map $K\colon\escript (H)\to\escript (H)$ is \textit{additive} if $K(a)\perp K(b)$ whenever $a\perp b$ and we have that $K(a+b)=K(a)+K(b)$. If $K$ is additive, then $K$ preserves order because if $a\le b$, then there exists a $c\in\escript (H)$ such that $a+c=b$ and we obtain
\begin{equation*}
K(a)\le K(a)+K(c)=K(a+c)=K(b)
\end{equation*}
If $K$ is additive and $K(I)=I$, then $K$ is a \textit{morphism} \cite{bgl95,dl70,gg02,gn01}.

A \textit{state} for $S$ is a positive operator $\rho\in\tscript (H)$ such that $\rmtr (\rho )=1$. We denote the set of states by $\sscript (H)$ and interpret $\rho\in\sscript (H)$ as an initial condition for the system $S$ \cite{bgl95,dl70}. We define the \textit{probibility that} $a\in\escript (H)$
\textit{is true} when $S$ is in the state $\rho$ by $\pscript _\rho (a)=\rmtr(\rho a)$. It follows that $\pscript _\rho (a')=1-\pscript _\rho (a)$ and
$a\le b$ if and only if $\pscript _\rho (a)\le\pscript _\rho (b)$ for all $\rho\in\sscript (H)$. An \textit{operation} on $H$ is a completely positive linear map $J\colon\lscript (H)\to\lscript (H)$ that is trace non-increasing for $\tscript (H)$ operators. We denote the set of operations on $H$ by
$\oscript (H)$. It can be shown \cite{bgl95,dl70,hz12,kra83} that every $J\in\oscript (H)$ has a \textit{Kraus decomposition}
$J(A)=\sum C_iAC_i^*$, $A\in\lscript (H)$, where $C_i\in\lscript (H)$ satisfy $\sum C_i^*C_i\le I$. This condition follows from the fact that for every $A\in\tscript (H)$ with $A\ge 0$ we have that
\begin{align*}
\rmtr\paren{\sum C_i^*C_iA}&=\sum\rmtr (C_i^*C_iA)=\sum\rmtr (C_iAC_i^*)\\
   &=\rmtr\paren{\sum C_iAC_i^*}\le\rmtr (A)
\end{align*}
holds if and only if $\sum C_i^*C_i\le I$. If an operation preserves the trace, it is called a \textit{channel} \cite{bgl95,dl70,gud320,nc00}. A
\textit{dual operation} on $H$ is a completely positive linear map $K\colon\lscript (H)\to\lscript (H)$ that satisfies
$K\colon\escript (H)\to\escript (H)$. It follows that $K|_{\escript (H)}$ is additive. We denote the set of dual operations on $H$ by
$\oscript ^*(H)$.
\begin{thm}    
\label{thm11}
If $J\colon\lscript (H)\to\lscript (H)$ is an operation, then there exists a unique $J^*\in\oscript ^*(H)$ such that
$\rmtr\sqbrac{\rho J^*(a)}=\rmtr\sqbrac{J(\rho )a}$ for all $a\in\escript (H)$, $\rho\in\sscript (H)$. Conversely, if $K\in\oscript ^*(H)$, then there exists a unique $J\in\oscript (H)$ such that $J^*=K$. Moreover, $J$ is a channel if and only if $J^*(I)=I$.
\end{thm}
\begin{proof}
Let $J\in\oscript (H)$ with Kraus decomposition $J(A)=\sum C_iAC_i^*$ where $\sum C_i^*C_i\le I$ and define $J^*(A)=\sum C_i^*AC_i$ for all $A\in\lscript (H)$. If $a\in\escript (H)$ and $\phi\in H$, since $0\le a\le I$, we have that
\begin{align*}
\elbows{\phi ,J^*(a)\phi}&=\elbows{\phi ,\sum C_i^*aC_i\phi}=\sum\elbows{C_i\phi ,aC_i\phi}\le\sum\elbows{C_i\phi ,C_i\phi}\\
   &=\elbows{\phi ,\sum C_i^*C_i\phi}\le\elbows{\phi ,\phi}
\end{align*}
Moreover, $\elbows{\phi ,J^*(a)\phi}\ge 0$ so $0\le J^*(a)\le I$ and we conclude that $J^*(a)\in\escript (H)$. Since $J^*$ also has a Kraus decomposition, it follows that $J^*\in\oscript ^*(H)$. The duality condition holds because
\begin{align}                
\label{eq11}
\rmtr\sqbrac{\rho J^*(a)}&=\rmtr\paren{\rho\sum C_i^*aC_i}=\sum\rmtr(\rho C_i^*aC_i)=\sum\rmtr (C_i\rho C_i^*a)\notag\\
   &=\rmtr\sqbrac{\sum C_i\rho C_i^*a}=\rmtr\sqbrac{J(\rho )a}
\end{align}
for all $a\in\escript (H)$, $\rho\in\sscript (H)$. To show that $J^*$ is unique, suppose $K\in\oscript ^*(H)$ satisfies
$\rmtr\sqbrac{\rho K(a)}=\rmtr\sqbrac{J(\rho )a}$ for all $a\in\escript (H)$, $\rho\in\sscript (H)$. Then
$\rmtr\sqbrac{\rho K(a)}=\rmtr\sqbrac{\rho J^*(a)}$ for all $a\in\escript (H)$, $\rho\in\sscript (H)$ so $K=J^*$. Conversely, let $K\in\oscript ^*(H)$ with Kraus decomposition $K(a)=\sum C_i^*aC_i$. Since $K\colon\escript (H)\to\escript (H)$ and $I\in\escript (H)$ we have that $K(I)\le I$. Hence,
\begin{equation*}
\sum C_i^*C_i=K(I)\le I
\end{equation*}
It follows that the map $J(A)=\sum C_iAC_i^*$ is an operation. As in \eqref{eq11} we have that
\begin{equation*}
\rmtr\sqbrac{\rho K(a)}=\rmtr\sqbrac{J(\rho )a}=\rmtr\sqbrac{\rho J^*(a)}
\end{equation*}
We conclude that $J^*=K$ and as before, $J$ is unique. If $J^*(I)=I$, then
\begin{equation*}
\rmtr\sqbrac{J(\rho )}=\rmtr\sqbrac{J(\rho )I}=\rmtr\sqbrac{\rho J^*(I)}=\rmtr (\rho )=1
\end{equation*}
for every $\rho\in\sscript (H)$ so $J$ is a channel. Conversely, if $J$ is a channel, then
\begin{equation*}
\rmtr\sqbrac{\rho J^*(I)}=\rmtr\sqbrac{J(\rho )I}=\rmtr\sqbrac{J(\rho )}=1
\end{equation*}
so $J^*(I)=I$.
\end{proof}

In the proof of Theorem~\ref{thm11}, we defined $J^*(A)=\sum C_i^*AC_i$, where $J$ has Kraus decomposition $J(A)=\sum C_iAC_i^*$. The Kraus operators $C_i$ are not unique and there can be many such operators \cite{hz12,nc00}. Suppose we have another Kraus decomposition $J(A)=\sum D_jAD_j^*$. By uniqueness, we conclude that $J^*(A)=\sum D_j^*AD_j$ so the form of the Kraus operators is immaterial. We say that an operation $J$ \textit{measures} an effect $a$ if 
\begin{equation*}
\rmtr\sqbrac{J(\rho )}=\rmtr (\rho a)=\pscript _\rho (a)
\end{equation*}
for every $\rho\in\sscript (H)$. We think of $J$ as an apparatus that can be employed to measure the effect $a$ \cite{gud120,gud220,gud320,gud21}. Then $\rmtr\sqbrac{J(\rho )}$ gives the probability that $a$ is true when the system $S$ is in the state
$\rho$. The operation $J$ gives more information than the effect $a$. If $\alpha\in\tscript (H)$ with $\alpha >0$ we define its corresponding state to be $\alphatilde =\alpha/\rmtr (\alpha )$. After an operation $J$ is performed, the state $\rho$ is \textit{updated} to the state
$(J\rho )^\sim$ \cite{gud120,gud220,gud320,gud21}. 

If $0\le\lambda _i\le 1$ with $\sum\lambda _i=1$ and $a_i\in\escript (H)$, then it is clear that $\sum\lambda _ia_i\in\escript (H)$ and if
$\rho _i\in\sscript (H)$ we have that $\sum\lambda _i\rho _i\in\sscript (H)$. We conclude that $\escript (H)$ and $\sscript (H)$ are closed under convex combinations and hence form convex sets. In a similar way, if $J_i\in\oscript (H)$, $K_i\in\iscript ^*(H)$, then
$\sum\lambda _iJ_i\in\oscript (H)$ and $\sum\lambda _iK_i\in\oscript ^*(H)$ so $\oscript (H)$ and $\oscript ^*(H)$ form convex sets. If
$J_1,J_2\in\oscript (H)$, we define their \textit{sequential product} $J_1\circ J_2\in\oscript (H)$ by $J_1\circ J_2(A)=J_2\paren{J_1(A)}$
\cite{gud120,gud220,gud320,gud21}. Physically, $J_1\circ J_2$ specifies the operation obtained by first employing the operation $J_1$ and then employing $J_2$. In a similar way, if $K_1,K_2\in\oscript ^*(H)$, their \textit{sequential product} $K_1\circ K_2\in\oscript ^*(H)$ is
$K_1\circ K_2(A)=K_2\paren{K_1(A)}$.

\begin{thm}    
\label{thm12}
{\rm{(i)}}\enspace An operation $J$ measures an unique effect given by $\jhat =J^*(I)$.
{\rm{(ii)}}\enspace If $0\le\lambda _i\le 1$ with $\sum\lambda _i=1$ and $J_i\in\oscript (H)$, then
$\paren{\sum\lambda _iJ_i}^*=\sum\lambda _iJ_i^*$ and $\paren{\sum\lambda _iJ_i}^\wedge =\sum\lambda _i\jhat _i$.
{\rm{(iii)}}\enspace If we also have $0\le\mu _j\le 1$ with $\sum\mu _j=1$ and $K_j\in\oscript (H)$ then
\begin{equation*}
\paren{\sum _i\lambda J_i}\circ\paren{\sum _j\mu _j K_j}=\sum _{i,j}\lambda _i\mu _jJ_i\circ K_j
\end{equation*}
and this result also holds if $J_i,K_j\in\oscript ^*(H)$.
{\rm{(iv)}}\enspace If $J,K\in\oscript (H)$ then $(J\circ K)^*=K^*\circ J^*$ and $(J\circ K)^\wedge =J^*(\khat\,)$.
{\rm{(v)}}\enspace The following statements are equivalent:
{\rm{(a)}}\enspace $J$ is a channel,
{\rm{(b)}}\enspace $\jhat =I$,
{\rm{(c)}}\enspace $J^*(I)=I$,
{\rm{(d)}}\enspace $(K\circ J)^\wedge=\khat$ for all $K\in\oscript (H)$.
\end{thm}
\begin{proof}
(i)\enspace Since
\begin{equation*}
\rmtr\sqbrac{J(\rho )}=\rmtr\sqbrac{J(\rho )I}=\rmtr\sqbrac{\rho J^*(I)}
\end{equation*}
for all $\rho\in\sscript (H)$, we conclude that $J$ measures $J^*(I)$. For uniqueness, if $J$ also measures $a$, then
\begin{equation*}
\rmtr (\rho a)=\rmtr\sqbrac{J(\rho )}=\rmtr\sqbrac{\rho J^*(I)}
\end{equation*}
for all $\rho\in\sscript (H)$ so $a=J^*(I)$.
(ii)\enspace Since
\begin{align*}
\rmtr\sqbrac{\rho\paren{\sum\lambda _iJ_i}^*(a)}&=\rmtr\sqbrac{\paren{\sum\lambda _iJ_i}(\rho )a}=\rmtr\sqbrac{\sum\lambda _iJ_i(\rho )a}\\
   &=\sum\lambda _i\rmtr\sqbrac{J_i(\rho )a}=\sum\lambda _i\rmtr\sqbrac{\rho J_i^*(a)}\\
   &=\rmtr\sqbrac{\rho\sum\lambda _iJ_i^*(a)}
\end{align*}
 for all $\rho\in\sscript (H)$, $a\in\escript (H)$, it follows that $\paren{\sum\lambda _iJ_i}^*=\sum\lambda _iJ_i^*$. We then obtain
 \begin{equation*}
 \paren{\sum\lambda _iJ_i}^\wedge =\paren{\sum\lambda _iJ_i}^*(I)=\sum\lambda _iJ_i^*(I)=\sum\lambda _i\jhat _i
\end{equation*}
which gives the result.
(iii)\enspace For all $A\in\lscript (H)$ we obtain
\begin{align*}
\paren{\sum _i\lambda _iJ_i}\circ\paren{\sum _j\mu _jK_j}&=\sum _j\mu _jK_j\paren{\sum _i\lambda _iJ_i(A)}\\
   &=\sum _{i,j}\lambda _i\mu _jK_j\paren{J_i(A)}=\sum _{i,j}\lambda _i\mu _jJ_i\circ K_j(A)
\end{align*}
and the result follows. The proof for $J_i,K_j\in\oscript ^*(H)$ is similar.
(iv)\enspace For all $\rho\in\sscript (H)$, $a\in\escript (H)$ we have that
\begin{align*}
\rmtr\sqbrac{\rho (J\circ K^*)(a)}&=\rmtr\sqbrac{(J\circ K)(\rho )a}=\rmtr\sqbrac{K\paren{J(\rho )}a}=\rmtr\sqbrac{J(\rho )K^*(a)}\\
    &=\rmtr\sqbrac{\rho J^*\paren{K^*(a)}}=\rmtr\sqbrac{\rho (K^*\circ J^*)(a)}
\end{align*}
Hence, $(J\circ K)^*=K^*\circ J^*$. It then follows from (i) that
\begin{equation*}
(J\circ K)^\wedge=(J\circ K)^*(I)=(K^*\circ J^*)(I)=J^*\paren{K^*(I)}=J^*(\khat\,)
\end{equation*}
(v)\enspace (a)$\Mapsto$(b)\enspace If $J$ is a channel, then for every $\rho\in\sscript (H)$ we have that
\begin{equation*}
\rmtr (\rho I)=1=\rmtr\sqbrac{J(\rho )I}=\rmtr\sqbrac{\rho J^*(I)}=\rmtr (\rho\jhat\,)
\end{equation*}
Hence, $\jhat =I$. 
(b)$\Leftrightarrow$(c)\enspace This follows from (i).
(c)$\Mapsto$(d)\enspace If $\jhat =J^*(I)=I$, applying (i) and (iv) gives
\begin{equation*}
(K\circ J)^\wedge =(K\circ J)^*(I)=(J^*\circ K^*)(I)=K^*\paren{J^*(I)}=K^*(\jhat\,)=K^*(I)=\khat
\end{equation*}
(d)$\Mapsto$(a)\enspace Suppose (d) holds and let $K$ be the identity channel $K(\rho )=\rho$ for all $\rho\in\sscript (H)$. Then $\khat =I$ so by (d) we have that
\begin{equation*}
\jhat =(K\circ J)^\wedge =\khat =I
\end{equation*}
We then obtain for all $\rho\in\sscript (H)$ that
\begin{equation*}
\rmtr\sqbrac{J(\rho )}=\rmtr\sqbrac{J(\rho )I}=\rmtr\sqbrac{\rho J^*(I)}=\rmtr (\rho\jhat\,)=\rmtr (\rho )=1
\end{equation*}
Hence, $J$ is a channel.
\end{proof}

The proof of the following result is similar to Theorem~\ref{thm12}(ii)

\begin{cor}    
\label{cor13}
If $J,K\in\oscript (H)$ and $J+K\in\oscript (H)$, then $(J+K)^*=J^*+K^*$ and $(J+K)^\wedge =\jhat +\khat$.
\end{cor}

If $a,b\in\escript (H)$ and $J\in\oscript (H)$ measures $a$ so that $\jhat =a$, we define the \textit{sequential product of} $a$ \textit{then}
$b$ \textit{relative to} $J$ by $a=\sqbrac{J}b=J^*(b)$. We interpret $a\sqbrac{J}b$ as the effect that results from first measuring $a$ with the operation $J$ and then measuring $b$. In this way, the measurement of $a$ can influence (or interfere with) $b$, but since we measure $b$ second, the measurement of $b$ does not influence $a$. An important point is that $a\sqbrac{J}b$ depends on $J$. As we shall see, there are many operations that measure $a$ so if $\khat =a$, then $a\sqbrac{J}b\ne a\sqbrac{K}b$, in general. Moreover, $a\sqbrac{J}b$ does not depend on an operation that measures $b$.

\begin{thm}    
\label{thm14}
{\rm{(i)}}\enspace If $\jhat =a$, $\khat =b$, then $a\sqbrac{J}b=(J\circ K)^\wedge$.
{\rm{(ii)}}\enspace $a\sqbrac{J}b\le a$ for all $a,b\in\escript (H)$.
{\rm{(iii)}}\enspace If $0\le\lambda\le 1$ and $\jhat =a$, then
$ (\lambda a)\sqbrac{\lambda J}b=\lambda\paren{a\sqbrac{J}b}=a\sqbrac{J}(\lambda b)$. 
{\rm{(iv)}}\enspace $a\sqbrac{J}I=a$ for all $a\in\escript (H)$.
{\rm{(v)}}\enspace $a\sqbrac{J}b'=a-J^*(a)$.
{\rm{(vi)}}\enspace If $\jhat _i=a_i$, $0\le\lambda _i\le 1$ with $\sum\lambda _i=1$ and $0\le\mu _j\le 1$ with $\sum\mu _j=1$, then for any $b_i\in\escript (H)$ we have that
\begin{equation*}
\paren{\sum\lambda _ia_i}\sqbrac{\sum\lambda _iJ_i}\paren{\sum\mu _jb_j}=\sum _{i,j}\lambda _i\mu _ja_i\sqbrac{J_i}b_j
\end{equation*}
\end{thm}
\begin{proof}
(i)\enspace By Theorem~\ref{thm12}(iv) we obtain
\begin{equation*}
a\sqbrac{J}b=J^*(b)=J^*(\khat\,)=(J\circ K)^\wedge
\end{equation*}
(ii)\enspace This follows from
\begin{equation*}
a\sqbrac{J}b=J^*(b)\le J^*(I)=\jhat = a
\end{equation*}
(iii)\enspace We have that
\begin{align*}
(\lambda a)\sqbrac{\lambda J}b=(\lambda J)^*(b)=\lambda J^*(b)=\lambda\paren{a\sqbrac{J}b}\\
\intertext{and}
\lambda J^*(b)=J^*(\lambda b)=a\sqbrac{J}(\lambda b)
\end{align*}
(iv)\enspace This follows from
\begin{equation*}
a\sqbrac{J}I=J^*(I)=\jhat =a
\end{equation*}
(v)\enspace We have that
\begin{equation*}
a\sqbrac{J}a'=a\sqbrac{J}(I-a)=J^*(I-a)=J^*(I)-J^*(a)=\jhat -\jhat (a)=a-J^*(a)
\end{equation*}
(vi)\enspace Applying Theorem~\ref{thm12}(ii) we obtain
\begin{align*}
\paren{\sum\lambda _ia_i}\sqbrac{\sum\lambda _iJ_i}\paren{\sum\mu _jb_j}&=\paren{\sum\lambda _iJ_i}^*\paren{\sum\mu _jb_j}\\
   &=\sum\lambda _jJ_i^*\paren{\sum\mu _jb_j}=\sum _{i,j}\lambda _i\mu _jJ_i^*(b_j)\\
   &=\sum _{i,j}\lambda _i\mu _ja_i\sqbrac{J_i}b_j
\end{align*}
\end{proof}

We end this section with some definitions suggested by the theory. We say that $a,b\in\escript (H)$ \textit{commute relative} to a subset
$\rscript\subseteq\oscript (H)$ if there exist operations $J,K\in\rscript$ such that
\begin{equation}                
\label{eq12}
a\sqbrac{J}b=b\sqbrac{K}a
\end{equation}
Of course, \eqref{eq12} is equivalent to $J^*(b)=K^*(a)$, where $\jhat =a$ and $\khat =b$. In particular, when \eqref{eq12} holds, then $a$ and $b$ commute relative to $\brac{J,K}\subseteq\oscript (H)$. When $\rscript =\oscript (H)$, we just say that $a$ and $b$ \textit{commute}. Since
\begin{equation*}
a\sqbrac{J}0=J^*(0)=0=0^*(a)=0\sqbrac{0}a
\end{equation*}
we conclude that any $a\in\escript (H)$ commutes with $0$. Similarly,
\begin{equation*}
a\sqbrac{J}I=J^*(I)=a=I^*(a)=I\sqbrac{I}a
\end{equation*}
So any $a\in\escript (H)$ commutes with $I$. Suppose $a$ commutes with $b$ so \eqref{eq12} holds. If $0\le\lambda\le 1$, then
\begin{equation*}
a\sqbrac{J}(\lambda b)=\lambda b\sqbrac{\lambda K}a
\end{equation*}
Hence, $a$ commutes with $b$. We do not know if the following conjecture holds.

\begin{conj}{1}  
If $a$ commutes with $b$ and $c$ where $b\perp c$, then $a$ commutes with $b+c$.
\end{conj}

Let $a,b\in\escript (H)$ and let $J,K\in\oscript (H)$ with $\jhat =a$, $\khat =a'$. We define the effect $b$ \textit{conditioned} by the effect $a$
\textit{relative} to $\brac{J,K}$ by
\begin{equation*}
\paren{b\ab{J,K}a}=a\sqbrac{J}b+a'\sqbrac{K}b
\end{equation*}
We interpret $\paren{b\ab{J,K}a}$ as the effect $b$ conditioned on whether $a$ is true or false as measured by the operations, $J,K$ respectively, In terms of probabilities, we have
\begin{align}                
\label{eq13}
\pscript _\rho\paren{b\ab{J,K}a}&=\rmtr\sqbrac{\rho J^*(b)}+\rmtr\sqbrac{\rho K^*(b)}
   =\rmtr\sqbrac{J(\rho )b}+\rmtr\sqbrac{K(\rho )b}\notag\\
   &=\pscript _\rho (a)\pscript _{\widetilde{J(\rho )}}(b)+\pscript _\rho (a')\pscript _{\widetilde{K(\rho )}}(b)
\end{align}
Equation \eqref{eq13} is a type of Bayes' rule where $\pscript _{\widetilde{J(\rho )}}$ is the probability that $b$ is true given that $a$ is true
and $\pscript _{\widetilde{K(\rho )}}(b)$ is the probability that $b$ is true given that $a$ is false. We say that $b$ \textit{is not influenced by} $a$ relative to $\brac{J,K}$ if $b=\paren{b\ab{J,K}a}$.

\section{L\"uders and Holevo Operations}  
The most important example of an operation is the L\"uders operation $L^a$, $a\in\escript (H)$, given by $L^a(A)=a^{1/2}Aa^{1/2}$. Since
\begin{equation*}
\rmtr\sqbrac{L^a(\rho )b}=\rmtr (a^{1/2}\rho a^{1/2}b)=\rmtr (\rho a^{1/2}ba^{1/2})=\rmtr\sqbrac{\rho (L^a)^*(b)}
\end{equation*}
we have that $(L^a)^*(b)=a^{1/2}ba^{1/2}=L^a(b)$ so $L^a$ is \textit{self-adjoint} in the sense that $L^a=(L^a)^*$. Moreover,
$(L^a)^\wedge =(L^a)^*(I)=a$ so $L^a$ measures $a$. In fact, $L^a$ is the unique L\"uders operation that measures $a$. An effect $a$ is
\textit{sharp} if $a$ is a projection. We denote the set of L\"uders operations by $\lscript$.

\begin{thm}    
\label{thm21}
{\rm{(i)}}\enspace $(L^a\circ J)^\wedge =a^{1/2}\jhat a^{1/2}$ for all $J\in\oscript (H)$.
{\rm{(ii)}}\enspace $a\sqbrac{L^a}b=a^{1/2}ba^{1/2}=L^a(b)$.
{\rm{(iii)}}\enspace $(J\circ L^a)^\wedge =J^*(a)$.
{\rm{(iv)}}\enspace $(L^a\circ L^b)^\wedge =a^{1/2}ba^{1/2}$.
{\rm{(v)}}\enspace $a$ commutes with $b$ relative to $\lscript$ if and only if $ab=ba$, that is, $a$ and $b$ commute in the usual operator sense.
{\rm{(vi)}}\enspace If $a$ is sharp, then $b$ is not influenced by $a$ relative to $\brac{L^a,L^{a'}}$ if and only if $ab=ba$.
\end{thm}
\begin{proof}
(i)\enspace By Theorem~\ref{thm12}(iv) we have that
\begin{equation*}
(L^a\circ J)^\wedge =(L^a)^*(\jhat\,)=a^{1/2}\jhat a^{1/2}
\end{equation*}
(ii)\enspace This follows from
\begin{equation*}
a\sqbrac{L^a}b=(L^a)^*(b)=a^{1/2}ba^{1/2}=L^a(b)
\end{equation*}
(iii)\enspace Applying Theorem~\ref{thm12}(iv) we obtain
\begin{equation*}
(J\circ L^a)^\wedge =J^*\sqbrac{(L^a)^\wedge}=J^*(a)
\end{equation*}
(iv) follows from (i).\enspace
(v)\enspace We have that $a$ commutes with $b$ relative to $\lscript$ if and only if
\begin{equation*}
a^{1/2}ba^{1/2} = a\sqbrac{L^a}b=b\sqbrac{L^b}a=b^{1/2}ab^{1/2}
\end{equation*}
which holds if and only if $ab=ba$ \cite{gn01}.
(vi)\enspace If $a$ is sharp then $a^{1/2}=a$ so $b$ is not influenced by $a$ relative to $\brac{L^a,L^{a'}}$ if and only if
\begin{equation*}
b=a\sqbrac{L^a}b+a'\sqbrac{L^{a'}}b=aba+a'ba'
\end{equation*}
Multiplying on left by $a$ gives $ab=aba$. Hence, $ab=(ab)^*=b^*a^*=ba$. Conversely, if $ab=ba$, then
\begin{equation*}
aba+a'ba'=ab+a'b=b\qedhere
\end{equation*}
\end{proof}

Theorem~\ref{thm21}(i) and (iii) show that $L^a\circ J$ measures $a^{1/2}\jhat a^{1/2}$ and $J\circ L^a$ measures $J^*(a)$ for all $J\in\oscript (H)$. We call $a\square b=a\sqbrac{L^a}b=a^{1/2}ba^{1/2}$ the \textit{standard sequential product} of $a$ and $b$ \cite{gg02,gn01,gud320,gud21}. Of course, if $J\ne L^a$ then $a\sqbrac{J}b\ne a^{1/2}ba^{1/2}$, in general. Theorem~\ref{thm14}(vi) shows that in a certain sense, a sequential product preserves convex combinations. This does not imply that when $0\le\lambda\le 1$ we have
\begin{equation*}
\sqbrac{\lambda a+(1-\lambda )b}\square c=\lambda a\square c+(1-\lambda )b\square c
\end{equation*}
which does not hold in general. In fact, we have that
\begin{equation*}
\sqbrac{\lambda a+(1-\lambda )b}{\square} c=\sqbrac{\lambda a+(1-\lambda )b}^{1/2}c\sqbrac{\lambda a+(1-\lambda )b}^{1/2}
\end{equation*}
On the other hand
\begin{equation*}
\lambda a\square c+(1-\lambda )b\square c=\lambda a^{1/2}ca^{1/2}+(1-\lambda )b^{1/2}cb^{1/2}
\end{equation*}

For $\alpha\in\sscript (H)$, $a\in\escript (H)$, we call $H_{(\alpha ,a)}(\rho )=\rmtr (\rho a)\alpha$ the
\textit{Holevo operation with state $\alpha$ and effect} $a$ \cite{hol94}. The next theorem shows that the sequential product of any operation with a Holevo operation is again a Holevo operation. It also shows that $\hhat _{(\alpha ,a)}=a$ for any $\alpha\in\sscript (H)$. This illustrates the fact that an effect can be measured by many operations. We denote the set of Holevo operations by $\hscript$.

\begin{thm}    
\label{thm22}
{\rm{(i)}}\enspace $H^*_{(\alpha ,a)}(b)=\rmtr (\alpha b)a$ and $\hhat _{(\alpha ,a)}=a$ for all $\alpha\in\sscript (H)$, $a,b\in\escript (H)$.
{\rm{(ii)}}\enspace $H_{(\alpha ,a)}\circ J=H_{({\widetilde{J\alpha}},\rmtr (J\alpha )a)}$ and $J\circ H_{(\alpha ,a)}=H_{(\alpha ,J^*(a))}$.
{\rm{(iii)}}\enspace $(H_{(\alpha ,a)}\circ J)^\wedge=\rmtr\sqbrac{J(\alpha )}a$ and $(J\circ H_{(\alpha ,a)})^\wedge =J^*(a)$.
{\rm{(iv)}}\enspace $H_{(\beta ,b)}\circ H_{(\alpha ,a)}=H_{(\alpha ,\rmtr (\beta a)b)}$.
{\rm{(v)}}\enspace $a\sqbrac{H_{(\alpha ,a)}}b=\sqbrac{\rmtr (\alpha b)}a$.
{\rm{(vi)}}\enspace $a$ commutes with $b$ relative to $\hscript$ if and only if there exists $\alpha ,\beta\in\sscript (H)$ such that
$\rmtr (\alpha b)a=\rmtr (\beta a)b$.
{\rm{(vii)}}\enspace $a$ does not influence $b$ relative to $\brac{H_{(\alpha ,a)},H_{(\beta ,a')}}$ if and only if
$b=\rmtr\sqbrac{(\alpha -\beta )b}a+\rmtr (\beta b)I$. In particular, if $\alpha =\beta$ then $b=\rmtr (\alpha b)I$.
{\rm{(viii)}}\enspace $\paren{b\ab{H_{(\alpha ,a)},H_{(\beta ,a')}}a}=\rmtr\sqbrac{(\alpha -\beta )b}a+\rmtr (\beta b)I$.
\end{thm}
\begin{proof}
(i)\enspace We have that
\begin{align*}
\rmtr\sqbrac{\rho H^*_{(\alpha ,a)}b}&=\rmtr\sqbrac{H_{(\alpha ,a)}(\rho )b}=\rmtr\sqbrac{\rmtr (\rho a)\alpha b}=\rmtr (\rho a)\rmtr (\alpha b)\\
   &=\rmtr\sqbrac{\rho\rmtr (\alpha b)a}
\end{align*}
It follows that $H^*_{(\alpha ,a)}(b)=\rmtr (\alpha b)a$. We conclude that $H_{(\alpha ,a)}$ measures the effect
\begin{equation*}
\hhat _{(\alpha ,a)}=H^*_{(\alpha ,a)}(I)=\rmtr (\alpha I)a=a
\end{equation*}
(ii)\enspace For all $\rho\in\sscript (H)$ obtain
\begin{align*}
(H_{(\alpha ,a)}\circ J)(\rho )&=J\sqbrac{H_{(\alpha ,a)}(\rho )}=J\sqbrac{\rmtr (\rho a)\alpha}=\rmtr (\rho a)J(\alpha )\\
   &=\rmtr\sqbrac{\rho\rmtr (J\alpha )a}{\widetilde{J\alpha}}=H_{({\widetilde{J\alpha}},\rmtr (J\alpha )a)}(\rho )
\end{align*}
and the result follows. Moreover, for all $\rho\in\sscript (H)$ we obtain
\begin{align*}
(J\circ H_{(\alpha ,a)})(\rho )&=H_{(\alpha ,a)}\sqbrac{J(\rho )}=\rmtr\sqbrac{J(\rho )a}\alpha =\rmtr\sqbrac{\rho J^*(a)}\alpha\\
   &=H_{(\alpha ,J^*(a))}(\rho )
\end{align*}
and the result follows.
(iii)\enspace These follow from (i) and (ii).
(iv)\enspace Applying (i) and (ii) gives
\begin{equation*}
H_{(\beta ,b)}\circ H_{(\alpha ,a)}=H_{(\alpha H^*_{(\beta ,b)}(a))}=H_{(\alpha ,\rmtr (\beta a)b)}
\end{equation*}
(v)\enspace Applying (i) gives
\begin{equation*}
a\sqbrac{H_{(\alpha ,a)}}b=H^*_{(\alpha ,a)}(b)=\rmtr (\alpha b)a
\end{equation*}
(vi)\enspace By (v) we have that $a\sqbrac{H_{(\alpha ,a)}}b=\rmtr (\alpha b)a$ and $b\sqbrac{H_{(\beta ,b)}}a=\rmtr (\beta a)b$. Hence,
$a\sqbrac{H_{(\alpha ,a)}}b=b\sqbrac{H_{(\beta ,b)}}a$ if and only if $\rmtr (\alpha b)a=\rmtr (\beta a)b$.
(vii)\enspace For all $\alpha ,\beta\in\sscript (H)$ we have by (v) that
\begin{align*}
a\sqbrac{H_{(\alpha ,a)}}b+&a'\sqbrac{H_{(\beta ,a')}}b=H^*_{(\alpha ,a)}(b)+H^*_{(\beta ,a')}(b)=\rmtr (\alpha b)a+\rmtr (\beta b)a'\\
    &=\rmtr (\alpha b)a+\rmtr (\beta b)I-\rmtr (\beta b)a=\rmtr\sqbrac{(\alpha -\beta )b}a+\rmtr (\beta b)I
\end{align*}
The result follows. (viii)\enspace This follows from (vii).
\end{proof}

Theorem~\ref{thm14}(iv) shows that $a\sqbrac{J}I=a$ for all $a\in\escript (H)$. We can use Holevo operations to show that
$I\sqbrac{J}a\ne a$, in general. Applying Theorem~\ref{thm22}(i) we have that
\begin{equation*}
I\sqbrac{H_{(\alpha ,I)}}a=H^*_{(\alpha ,I)}(a)=\rmtr (\alpha a)I\ne a
\end{equation*}
in general.

\section{Instruments and Observables}  
We now extend our previous work to the theory of instruments and observables. If $(\Omega _\iscript ,\fscript _\iscript )$ is a measurable space, an \textit{instrument} on $H$ with \textit{outcome space} $(\Omega _\iscript ,\fscript _\iscript )$ is an operation-valued measure on
$\fscript _\iscript$. That is, $\Delta\mapsto\iscript (\Delta )\in\oscript (H)$ is countably additive relative to a suitable topology and
$\iscript (\Omega _\iscript )=\iscriptbar$ is a channel \cite{bgl95,dl70,hz12,nc00}. We denote the set of instruments on $H$ by $\inset (H)$.
We interpret an instrument as an apparatus that can be employed to perform measurements. Then $\iscript (\Delta )$ is the operation that results when a measurement of $\iscript$ gives an outcome in $\Delta$. For any $\rho\in\sscript (H)$, we call
$\Phi _\rho ^\iscript (\Delta )=\rmtr\sqbrac{\iscript (\Delta )(\rho )}$ the \textit{distribution} of $\iscript$ in the state $\rho$ and interpret
$\Phi _\rho ^\iscript (\Delta )$ as the probability that a measurement of $\iscript$ results in an outcome in $\Delta$ when the system is in the state $\rho$. Notice that since $\iscriptbar =\iscript (\Omega _\iscript )$ is a channel, we have that
\begin{equation*}
\Phi _\rho ^\iscript (\Omega _\iscript )=\rmtr\sqbrac{\,\iscriptbar (\rho )}=1
\end{equation*}
so $\Phi _\rho ^\iscript$ is a probability measure for every $\rho\in\sscript (H)$. If $\jscript$ is another instrument with outcome space
$(\Omega _\jscript ,\fscript _\jscript )$, their \textit{sequential product} $\jscript\circ\jscript$ of $\iscript$ then $\jscript$ is the instrument with outcome space $(\Omega _\iscript\times\Omega _\jscript ,\fscript _\iscript\times\fscript _\jscript )$ that satisfies
\begin{equation*}
(\iscript\circ\jscript )(\Delta\times\Gamma )(\rho )=\jscript (\Gamma )\paren{\iscript (\Delta )(\rho )}
\end{equation*}
for all $\Delta\in\fscript _\iscript$, $\Gamma\in\fscript _\jscript$, $\rho\in\sscript (H)$. The \textit{joint distribution} satisfies
\begin{equation*}
\Phi _\rho ^{\iscript\circ\jscript}(\Delta\times\Gamma )=\rmtr\sqbrac{(\iscript\circ\jscript )(\Delta\times\Gamma )(\rho )}
  =\rmtr\sqbrac{\jscript (\Gamma )\paren{\iscript (\Delta )}(\rho )}
\end{equation*}
We define $\jscript$ \textit{conditioned by} $\iscript$ to be the instrument $(\jscript\mid\iscript )$ with outcome space
$(\Omega _\jscript ,\fscript _\jscript )$ given by
\begin{equation*}
(\jscript\mid\iscript )(\Gamma )(\rho )=\jscript (\Gamma )\sqbrac{\,\iscriptbar (\rho )}
    =\jscript (\Gamma )\sqbrac{\iscript (\Omega _\iscript )(\rho )}
\end{equation*}

If $\iscript\in\inset (H)$ we have that $\iscript (\Delta )\in\oscript (H)$ and hence $\iscript (\Delta )^*\in\oscript ^*(H)$ for all
$\Delta\in\fscript _\iscript$. We call $\iscript ^*(\Delta )=\iscript (\Delta )^*$ a \textit{dual instrument}. Thus, $\iscript ^*$ is a dual
operation-valued measure on $(\Omega _\iscript ,\fscript _\iscript )$ satisfying $\iscript ^*(\Delta )\colon\escript (H)\to\escript (H)$ for all
$\Delta\in\fscript _\iscript$ and by Theorem~\ref{thm11}
\begin{equation*}
\iscript ^*(\Omega _\iscript )(I)=\iscriptbar (I)=I
\end{equation*}
Moreover, $\iscript ^*$ is the unique dual instrument satisfying
\begin{equation}                
\label{eq31}
\rmtr\sqbrac{\rho\iscript ^*(\Delta )(a)}=\rmtr\sqbrac{\iscript (\Delta )(\rho )a}
\end{equation}
for all $\rho\in\sscript (H)$, $\Delta\in\fscript _\iscript$, $a\in\escript (H)$. We denote the set of dual instruments by $\inset^*(H)$.

If $(\Omega _A,\fscript _A)$ is a measurable space, an \textit{observable} $A$ on $H$ with \textit{outcome space} $(\Omega _A,\fscript _A)$ is an effect-valued measure on $\fscript _A$ satisfying $A(\Omega _A)=I$ \cite{bgl95,dl70,gud120, hz12,nc00}. We denote the set of observables on $H$ by $\obset (H)$. If $A\in\obset (H)$, we interpret $A(\Delta )$ as the effect resulting from $A$ having an outcome in
$\Delta\in\fscript _A$ when $A$ is measured. The probability that $A$ results in an outcome in $\Delta$ when the system is in the state
$\rho\in\sscript (H)$ is given by $\Phi _\rho ^A(\Delta )=\rmtr\sqbrac{\rho A(\Delta )}$ and $\Phi _\rho ^A$ is the \textit{distribution} of $A$ in the state $\rho$. If $\iscript\in\inset (H)$, the unique observable $\iscripthat\in\oscript (H)$ \textit{measured} by $\iscript$ has outcome space
$(\Omega _\iscript ,\fscript _\iscript )$ and satisfies \cite{dl70,hz12,nc00}
\begin{equation}                
\label{eq32}
\Phi _\rho ^{\iscripthat}(\Delta )=\rmtr\sqbrac{\rho\iscripthat (\Delta )}=\rmtr\sqbrac{\iscript (\Delta )(\rho )}=\Phi _\rho ^\iscript (\Delta )
\end{equation}
Applying \eqref{eq31} and \eqref{eq32} we obtain
\begin{equation*}
\rmtr\sqbrac{\rho\iscripthat (\Delta )}=\rmtr\sqbrac{\rho\iscripthat (\Delta )I}=\rmtr\sqbrac{\iscript (\Delta )(\rho )I}
   =\rmtr\sqbrac{\rho\iscript ^*(\Delta )(I)}
\end{equation*}
for all $\rho\in\sscript (H)$. Hence, for all $\Delta\in\fscript _\iscript$ we have
\begin{equation}                
\label{eq33}
\iscripthat (\Delta )=\iscript ^*(\Delta )(I)
\end{equation}
As with operations, although $\iscript\in\inset (H)$ measures an unique observable $\iscripthat$, an observable is measured by many instruments. Moreover, we interpret an instrument $\iscript$ as an apparatus that can be employed to measure the observable $\iscripthat$. The next result follows from Theorem~\ref{thm12}.

\begin{thm}    
\label{thm31}
{\rm{(i)}}\enspace If $0\le\lambda _i\le 1$ with $\sum\lambda _i=1$ and $\iscript _i\in\inset (H)$, then $\sum\lambda _i\iscript _i\in\inset (H)$,
$\paren{\sum\lambda _i\iscript _i}^*=\sum\lambda _i\iscript _i^*$ and
$\paren{\sum\lambda _i\iscript _i}^\wedge =\sum\lambda _i\iscripthat _i$.
{\rm{(ii)}}\enspace If we also have $0\le\mu _j\le 1$ with $\sum\mu _j=1$ and $\jscript _j\in\inset (H)$, then
\begin{equation*}
\paren{\sum\lambda _i\iscript _i}\circ\paren{\sum\mu _j\jscript _j}=\sum _{i,j}\lambda _i\mu _j\iscript _i\circ\jscript _j
\end{equation*}
and a similar result holds where $\iscript _i,\jscript _j\in\inset ^*(H)$.
{\rm{(iii)}}\enspace If $\iscript ,\jscript\in\inset (H)$, then $(\iscript\circ\jscript )^*=\jscript ^*\circ\iscript ^*$ and
$(\iscript\circ\jscript )^\wedge =\iscript ^*(\,\jscripthat )$.
\end{thm}

Let $A,B\in\obset (H)$ and let $\iscript\in\inset (H)$ satisfy $\iscripthat =A$. We define the \textit{sequential product of} $A$
\textit{then} $B$ \textit{relative to} $\iscript$ as the observable with outcome space $(\Omega _A\times\Omega _B$,
$\fscript _A\times\fscript _B)$ given by $A\sqbrac{\iscript}B=\iscript ^*(B)$. This is shorthand notation for
\begin{equation}                
\label{eq34}
A\sqbrac{\iscript}B(\Delta\times\Gamma)=\iscript ^*(B)(\Delta\times\Gamma )=\iscript ^*(\Delta )\paren{B(\Gamma )}
   =\iscript (\Delta )^*\paren{B(\Gamma )}
\end{equation}
Notice that $A\sqbrac{\iscript}B$ depends on the instrument $\iscript$ that measures $A$, but does not depend on the instrument measuring $B$. This is because $B$ is measured second so its measurement cannot influence the $A$ measurement. Applying \eqref{eq31} and \eqref{eq34}, the distribution of $A\sqbrac{\iscript}B$ satisfies
\begin{align}                
\label{eq35}
\Phi _\rho ^{A\sqbrac{\iscript}B}(\Delta\times\Gamma )&=\rmtr\sqbrac{\rho A\sqbrac{\iscript}B(\Delta\times\Gamma )}
  =\rmtr\sqbrac{\rho\iscript (\Delta )^*\paren{B(\Gamma )}}\notag\\
  &=\rmtr\sqbrac{\iscript (\Delta )(\rho )B(\Gamma )}
\end{align}
It follows from \eqref{eq35} that
\begin{equation*}
\Phi _\rho ^{A\sqbrac{\iscript}B}(\Delta\times\Omega _B)=\rmtr\sqbrac{\iscript (\Delta )(\rho )}=\Phi _\rho^A(\Delta )
\end{equation*}
for all $\rho\in\sscript (H)$, $\Delta\in\fscript _A$. We define $B$ \textit{conditioned by} $A$ \textit{relative to} $\iscript$ as the observable with outcomes space $(\Omega _B,\fscript _B)$ given by
\begin{equation*}
\paren{B\mid\iscript\mid A}(\Gamma )=\iscriptbar\,^*\paren{B(\Gamma )}
\end{equation*}
for all $\Gamma\in\fscript _B$. The distribution of $\paren{B\mid\iscript\mid A}$ becomes
\begin{align}                
\label{eq36}
\Phi _\rho ^{\paren{B\mid\iscript\mid A}}(\Gamma )&=\rmtr\sqbrac{\rho\iscriptbar\,^*\paren{B(\Gamma )}}
   =\rmtr\sqbrac{\,\iscriptbar (\rho )B(\Gamma )}=\rmtr\sqbrac{\iscript (\Omega _\iscript )(\rho )B(\Gamma )}\notag\\
   &=\Phi _\rho ^{A\sqbrac{\iscript}B}(\Omega _\iscript\times\Gamma )
\end{align}

Notice that this idea has already been presented in the quantum formalism when we consider the updated state after the measurement of $A$ results in an outcome in $\Delta$. This updated state depends on the instrument $\iscript$ employed to measure $A$ and is given by
\begin{equation*}
\rho\mapsto\iscript (\Delta )\rho/\rmtr\sqbrac{\iscript (\Delta )(\rho )}=\sqbrac{\iscript (\Delta )(\rho )}^\sim
\end{equation*}
Using a different instrument to measure $A$ results in a different updated state in general. Even though $A\sqbrac{\iscript}B$ and
$\paren{B\mid\iscript\mid A}$ do not depend on the instrument $\jscript$ used to measure $B$, the next result gives an expression involving
$\jscript$.

\begin{lem}    
\label{lem32}
Let $\iscript ,\jscript\in\inset (H)$ satisfy $\iscripthat =A$, $\jscripthat =B$.\newline
{\rm{(i)}}\enspace $A\sqbrac{\iscript}B=(\iscript\circ\jscript )^\wedge$.
{\rm{(ii)}}\enspace $\paren{B\mid\iscript\mid A}=\paren{\jscript\mid\iscript}^\wedge$.
\end{lem}
\begin{proof}
(i)\enspace By Theorem~\ref{thm31}(iii) we obtain
\begin{equation*}
(\iscript\circ\jscript )^\wedge =\iscript ^*(\,\jscripthat\,)=\iscript ^*(B)=A\sqbrac{\iscript}B
\end{equation*}
(ii)\enspace For all $\rho\in\sscript (H)$, $\Gamma\in\fscript _B$ we have
\begin{align*}
\rmtr\sqbrac{\rho (\jscript\mid\iscript )^\wedge (\Gamma )}&=\rmtr\sqbrac{\rho (\jscript\mid\iscript )^*(\Gamma )(I)}
   =\rmtr\sqbrac{(\jscript\mid\iscript )(\Gamma )(\rho )I}\\
   &=\rmtr\brac{\jscript (\Gamma )\sqbrac{\,\iscriptbar (\rho )}I}=\rmtr\sqbrac{\,\iscriptbar (\rho )\jscript ^*(\Gamma )(I)}\\
   &=\rmtr\sqbrac{\iscriptbar (\rho )\,\jscripthat (\Gamma )}=\rmtr\sqbrac{\,\iscriptbar (\rho )B(\Gamma )}
   =\rmtr\sqbrac{\rho\iscriptbar\,^*\paren{B(\Gamma )}}
\end{align*}
Hence,
\begin{equation*}
(\jscript\mid\iscript )^\wedge (\Gamma )=\iscriptbar\,^*\paren{B(\Gamma )}=\paren{B\mid\iscript\mid A}(\Gamma )
\end{equation*}
for all $\Gamma\in\fscript _B$ so $(\jscript\mid\iscript )^\wedge =\paren{B\mid\iscript\mid A}$.
\end{proof}

If $\mu$ is a probability measure on $(\Omega ,\fscript )$ we call $\iscript _\mu(\Delta )(\rho )=\mu (\Delta )\rho$ for $\Delta\in\fscript$ an
\textit{identity instrument with measure} $\mu$. Similarly, we define the \textit{identity observable with measure} $\mu$ as $A_\mu (\Delta )=\mu (\Delta )I$ for $\Delta\in\fscript$. These are the simplest types of instruments and observables. The next theorem illustrates this theory in terms of these simple types. We first need an elementary lemma.

\begin{lem}    
\label{lem33}
If $A,B\in\obset (H)$, then $A\sqbrac{\iscript}B(\Omega _A\times\Gamma )=\paren{B\mid\iscript\mid A}(\Gamma )$ and
$A\sqbrac{\iscript}B(\Delta\times\Omega _B)=A(\Delta)$.
\end{lem}
\begin{proof}
For all $\Gamma\in\fscript _B$ we obtain
\begin{equation*}
A\sqbrac{\iscript}B(\Omega _A\times\Gamma )=\iscript (\Omega _\iscript )^*\paren{B(\Gamma )}=\iscriptbar\,^*\paren{B(\Gamma )}
   =\paren{B\mid\iscript\mid A}(\Gamma )
\end{equation*}
Moreover, for all $\Delta\in\fscript _A$ we obtain
\begin{equation*}
A\sqbrac{\iscript}B(\Delta\times\Omega _B)=\iscript (\Delta )^*\paren{B(\Omega _B)}=\iscript (\Delta )^*I=\iscripthat (\Delta )=A(\Delta )
\qedhere
\end{equation*}
\end{proof}

\begin{thm}    
\label{thm34}
Let $\iscript _\mu$ be the identity instrument with measure $\mu$.
{\rm{(i)}}\enspace $\iscript _\mu ^*(\Delta )(a)=\mu (\Delta )(a)$ for all $\Delta\in\fscript$, $a\in\escript (H)$ and
$\iscripthat _\mu (\Delta )=\mu (\Delta )I$ is the identity observable with measure $\mu$.
{\rm{(ii)}}\enspace If $A=\iscripthat _\mu$ and $B=\jscripthat$, then
\begin{equation*}
A\sqbrac{\iscript _\mu}B(\Delta\times\Gamma )=B\sqbrac{\jscript}A(\Gamma\times\Delta )=\mu (\Delta )B(\Gamma )
\end{equation*}
{\rm{(iii)}}\enspace If $A=\iscripthat _\mu$ and $B=\jscripthat$, then $\paren{B\mid\iscript _\mu\mid A}=B$ and
$\paren{A\mid\jscript\mid B}=A$.
{\rm{(iv)}}\enspace If $A=\iscripthat _\mu$ and $B=\iscripthat _\nu$, then $A\sqbrac{\iscript _\mu}B$ is the identity observable with measure
$\mu\times\nu$.
{\rm{(v)}}\enspace If $\jscript\in\inset (H)$, then $\paren{\jscript\mid\iscript _\mu}=\jscript$,
$\paren{\iscript _\mu\mid\jscript}(\Delta )(\rho )=\mu (\Delta )\sqbrac{\,\jscripthat (\rho )}$, $(\jscript\mid\iscript _\mu )^\wedge =\jscripthat$ and
$(\iscript _\mu\mid\jscript )^\wedge =\iscripthat _\mu$.
\end{thm}
\begin{proof}
(i)\enspace For all $\rho\in\sscript (H)$, $a\in\escript (H)$, $\Delta\in\fscript$ we have
\begin{equation*}
\sqbrac{\rho\iscript _\mu ^*(\Delta )(a)}=\rmtr\sqbrac{\iscript _\mu (\Delta )(\rho )a}=\rmtr\sqbrac{\mu (\Delta )\rho a}
   =\rmtr\sqbrac{\rho\mu (\Delta )a}
\end{equation*}
Hence, $\iscript _\mu ^*(\Delta )(a)=\mu (\Delta )a$. It follows that
\begin{equation*}
\iscripthat _\mu (\Delta )=\iscript _\mu ^*(\Delta )(I)=\mu (\Delta )I
\end{equation*}
(ii)\enspace Since $A\sqbrac{\iscript _\mu}B=\iscript _\mu ^*(B)$, applying (i) we obtain
\begin{equation*}
A\sqbrac{\iscript _\mu}B(\Delta\times\Gamma )=\iscript _\mu ^*(\Delta )\paren{B(\Gamma )}=\mu (\Delta )B(\Gamma )
\end{equation*}
Since $B\sqbrac{\jscript}A=\jscript ^*(A)$ we obtain
\begin{align*}
B\sqbrac{\jscript}A(\Gamma\times\Delta )&=\jscript ^*(\Gamma )\paren{A(\Delta )}=\jscript ^*(\Gamma )\paren{\mu (\Delta )I}\\
   &=\mu (\Delta )\jscript ^*(\Gamma )(I)=\mu (\Delta )\jscripthat (\Gamma )=\mu (\Delta )B(\Gamma )
\end{align*}
(iii)\enspace Applying (ii) and Lemma~\ref{lem33} gives
\begin{equation*}
\paren{B\mid\iscript _\mu\mid A}(\Gamma )=A\sqbrac{\iscript _\mu}B(\Omega _A\times\Gamma )=\mu (\Omega _\iscript )B(\Gamma )
   =B(\Gamma )
\end{equation*}
Hence, $\paren{B\mid\iscript _\mu\mid A}=B$. Since $\jscripthat =B$, applying Lemma~\ref{lem33} gives
\begin{equation*}
\paren{A\mid\jscript\mid B}(\Delta )=B\sqbrac{\jscript }A(\Omega _B\times\Delta )=\mu (\Delta )B(\Omega _B )=\mu (\Delta )I=A(\Delta )
\end{equation*}
Hence, $\paren{A\mid\jscript\mid B}=A$.
(iv)\enspace Applying (ii) gives
\begin{align*}
A\sqbrac{\iscript _\mu}B(\Delta\times\Gamma )&=\mu (\Delta )B(\Gamma )=\mu (\Delta )\iscripthat _\nu (\Gamma )\\
   &=\mu (\Delta )\nu (\Gamma )I=(\mu\times\nu )(\Delta\times\Gamma )I
\end{align*}
and the result follows.
(v)\enspace For all $\Gamma\in\fscript _\jscript$, $\rho\in\sscript (H)$ we obtain
\begin{equation*}
\paren{\jscript\mid\iscript _\mu}(\Gamma )(\rho )=\jscript (\Gamma )\sqbrac{\,\iscriptbar _\mu (\rho )}=\jscript (\Gamma )(\rho )
\end{equation*}
Hence, $\paren{\jscript\mid\iscript _\mu}=\jscript$. For all $\Delta\in\fscript$, $\rho\in\sscript (H)$ we obtain
\begin{equation*}
\paren{\iscript _\mu\mid\jscript}(\Delta )(\rho )=\iscript _\mu (\Delta )\sqbrac{\,\jscriptbar (\rho )}=\mu (\Delta )\sqbrac{\,\jscriptbar (\rho )}
\end{equation*}
Moreover, for all $\Gamma\in\fscript _\jscript$ we have
\begin{equation*}
\paren{\jscript\mid\iscript _\mu}^\wedge (\Gamma )=\iscriptbar _\mu ^*\paren{\,\jscripthat (\Gamma )}=\jscripthat (\Gamma )
\end{equation*}
Hence, $\paren{\jscript\mid\iscript _\mu}^\wedge =\jscripthat$. Finally, we have for all $\Delta\in\fscript$ that
\begin{align*}
\paren{\iscript _\mu\mid\jscript}^\wedge (\Delta )&=\jscriptbar\,^*\paren{\,\iscripthat _\mu (\Delta )}=\jscriptbar\,^*\sqbrac{\mu (\Delta )I}\\
   &=\mu (\Delta )\jscriptbar\,^*(I)=\mu (\Delta )I=\iscripthat _\mu (\Delta )
\end{align*}
so $\paren{\iscript _\mu\mid\jscript}^\wedge =\iscripthat _\mu$.
\end{proof}
We can extend the definition of a Holevo operation to a Holevo instrument as follows. A \textit{Holevo instrument with state} $\alpha$
\textit{and observable} $A$ has the form $\hscript _{(\alpha ,A)}(\Delta )(\rho )=\rmtr\sqbrac{\rho A(\Delta )}\alpha$ for all
$\Delta\in\Omega _A$.

\begin{thm}    
\label{thm35}
Let $\hscript _{(\alpha ,A)}$ be a Holevo instrument.
{\rm{(i)}}\enspace $\hscript _{(\alpha ,A)}^*(\Delta )(a)=\rmtr (\alpha a)A(\Delta )$ for all $\Delta\in\fscript _A$, $a\in\escript (H)$ and
$\hscripthat _{(\alpha ,A)}=A$.
{\rm{(ii)}}\enspace $A\sqbrac{\hscript _{(\alpha ,A)}}B(\Delta\times\Gamma )=\rmtr\sqbrac{\alpha B(\Gamma )}A(\Delta )$.
{\rm{(iii)}}\enspace $\paren{B\mid\hscript _{(\alpha ,A)}\mid A}(\Gamma )=\rmtr\sqbrac{\alpha B(\Gamma )}I$ which is an identity observable.
\end{thm}
\begin{proof}
(i)\enspace For every $\rho\in\sscript (H)$, $\Delta\in\fscript _A$, $a\in\escript (H)$ we obtain
\begin{align*}
\rmtr\sqbrac{\rho\hscript _{(\alpha ,A)}^*(\Delta )(a)}&=\rmtr\sqbrac{\hscript _{(\alpha ,A)}(\Delta )(\rho )a}
   =\rmtr\brac{\rmtr\sqbrac{\rho A(\Delta )}\alpha a}\\
   &=\rmtr\sqbrac{\rho A(\Delta )}\rmtr (\alpha a)=\rmtr\brac{\rho\rmtr (\alpha a)A(\Delta )}
\end{align*}
Hence, $\hscript _{(\alpha ,A)}^*(\Delta )(a)=\rmtr (\alpha a)A(\Delta )$. Moreover,
\begin{equation*}
\hscripthat _{(\alpha ,A)}(\Delta )=\hscript _{(\alpha ,A)}^*(\Delta )I=A(\Delta )
\end{equation*}
for all $A\in\fscript _A$ so $\hscripthat _{(\alpha ,A)}=A$.
(ii)\enspace Applying (i) we have
\begin{align*}
A\sqbrac{\hscript _{(\alpha ,A)}}B(\Delta\times\Gamma )&=\hscript _{(\alpha ,A)}^*(B)(\Delta\times\Gamma )
   =\hscript _{(\alpha ,A)}^*(\Delta )\paren{B(\Gamma )}\\
   &=\rmtr\sqbrac{\alpha B(\Gamma )}A(\Delta )
\end{align*}
(iii)\enspace Applying Lemma~\ref{lem33} and (ii) give
\begin{align*}
\paren{B\mid\hscript _{(\alpha ,A)}\mid A}(\Gamma )&=A\sqbrac{\hscript _{(\alpha ,A)}}B(\Omega _A\times\Gamma )
   =\rmtr\sqbrac{\alpha B(\Gamma )}A(\Omega _A)\\
   &=\rmtr\sqbrac{\alpha B(\Gamma )}I\qedhere
\end{align*}
\end{proof}

An instrument $\iscript$ is \textit{state constant} if $\iscript (\Delta )(\rho _1)=\iscript (\Delta )(\rho _2)$ for all $\rho _1,\rho _2\in\sscript (H)$,
$\Delta\in\fscript _\iscript$. If $\jscript\in\inset (H)$, $\alpha\in\sscript (H)$, we define the $\alpha$-\textit{state constant} instrument
$\jscript _\alpha$ by $\jscript _\alpha (\Delta )(\rho )=\jscript (\Delta )(\alpha )$ for all $\Delta\in\fscript _\jscript$, $\rho\in\sscript (H)$. It follows that $\iscript$ is state constant if and only if $\iscript =\jscript _\alpha$ for some $\jscript\in\inset (H)$, $\alpha\in\sscript (H)$. For example, the
$\alpha$-state constant instrument $\sqbrac{\hscript _{(\beta ,A)}}_\alpha$ is given by
\begin{equation*}
\sqbrac{\hscript _{(\beta ,A)}}_\alpha (\Delta )(\rho )=\hscript _{(\beta ,A)}(\Delta )(\alpha )=\rmtr\sqbrac{\alpha A(\Delta )}\beta
\end{equation*}
for all $\Delta\in\fscript _A$, $\rho\in\sscript (H)$. Notice that $\jscript _\alpha$ can be extended by linearity to all $\tscript (H)$.

\begin{thm}    
\label{thm36}
If $\iscript ,\jscript\in\inset (H)$, $\alpha\in\sscript (H)$, the following statements hold.
{\rm{(i)}}\enspace $\jscript _\alpha ^*(\Delta )(a)=\rmtr\sqbrac{\jscript (\Delta )(\alpha )a}I$ and $\jscripthat _\alpha$ is the identity observable
$\jscripthat _\alpha (\Delta )=\rmtr\sqbrac{\jscript (\Delta )(\alpha )}I$.
{\rm{(ii)}}\enspace $(\iscript\mid\jscript _\alpha )=\iscript _{\jscriptbar (\alpha )}$, $(\jscript _\alpha\mid\iscript )=\jscript _\alpha$.
{\rm{(iii)}}\enspace $(\iscript\mid\jscript _\alpha )^\wedge$ is the identity observable,
\begin{equation*}
\paren{\iscript\mid\jscript _\alpha}^\wedge (\Delta )=\rmtr\sqbrac{\,\jscriptbar (\alpha )\iscripthat (\Delta )}I
\end{equation*}
and $(\jscript _\alpha\mid\iscript )^\wedge =\jscripthat _\alpha$.
{\rm{(iv)}}\enspace If $A=\jscripthat _\alpha$, then $A\sqbrac{\jscript _\alpha}B$ is the identity observable with measure
$\mu (\Delta\times\Gamma )=\rmtr\sqbrac{\jscript (\Delta )(\alpha )B(\Gamma )}$ and $(B\mid\jscript _\alpha\mid A)$ is the identity observable with measure $\rmtr\sqbrac{\,\jscriptbar (\alpha )B(\Gamma )}$.
{\rm{(v)}}\enspace $\paren{\jscript\mid\hscript _{(\alpha ,A)}}=\jscript _\alpha$.
\end{thm}
\begin{proof}
(i)\enspace For all $\rho\in\sscript (H)$, $\Delta\in\fscript _\jscript$, $a\in\escript (H)$ we have that
\begin{equation*}
\rmtr\sqbrac{\rho\jscript _\alpha ^*(\Delta )(a)}=\rmtr\sqbrac{\jscript _\alpha (\Delta )(\rho )(a)}=\rmtr\sqbrac{\jscript (\Delta )(\alpha )a}
  =\rmtr\brac{\rho\rmtr\sqbrac{\jscript (\Delta )(\alpha )a}I}
\end{equation*}
Hence, $\jscript _\alpha ^*(\Delta )=\rmtr\sqbrac{\jscript (\Delta )(\alpha )a}I$. Moreover,
\begin{equation*}
\jscripthat _\alpha (\Delta )=\jscript _\alpha ^*(\Delta )(I)=\rmtr\sqbrac{\jscript (\Delta )(\alpha )}I(
\end{equation*}
(ii)\enspace For all $\Delta\in\fscript _\iscript$, $\rho\in\sscript (H)$ we have
\begin{equation*}
(\iscript\mid\jscript _\alpha )(\Delta )(\rho )=\iscript (\Delta )\sqbrac{\,\jscriptbar _\alpha (\rho )}=\iscript (\Delta )\sqbrac{\,\jscriptbar (\alpha )}
   =\iscript _{\jscriptbar (\alpha )}(\Delta )(\rho )
\end{equation*}
Hence, $(\iscript\mid\jscript _\alpha )=\iscript _{\jscriptbar (\alpha )}$. Moreover, for all $\Delta\in\fscript _\jscript$, $\rho\in\sscript (H)$ we obtain
\begin{equation*}
(\jscript _\alpha\mid\iscript )(\Delta )(\rho )=\jscript _\alpha (\Delta )\sqbrac{\,\iscriptbar (\rho )}=\jscript (\Delta )(\alpha )
   =\jscript _\alpha (\Delta )(\rho )
\end{equation*}
Thus, $(\jscript _\alpha\mid\iscript )=\jscript _\alpha$.
(iii)\enspace For all $\Delta\in\fscript _\iscript$ we obtain
\begin{equation*}
(\iscript\mid\jscript _\alpha )^\wedge (\Delta )=\jscriptbar _\alpha ^*\sqbrac{\,\iscripthat (\Delta )}
   =\rmtr\sqbrac{\,\jscriptbar (\alpha )\iscripthat (\Delta )}I
\end{equation*}
Applying (ii) gives $(\jscript _\alpha\mid\iscript )^\wedge =\jscripthat _\alpha$.
(iv)\enspace For $\Delta\in\fscript _A$, $\Gamma\in\fscript _B$, applying (i) we obtain
\begin{equation*}
A\sqbrac{\jscript _\alpha}B(\Delta\times\Gamma )=\jscript _\alpha ^*(\Delta )\paren{B(\Gamma )}
   =\rmtr\sqbrac{\jscript (\Delta )(\alpha )B(\Gamma )}I
\end{equation*}
Moreover, for all $\Gamma\in\fscript _B$ we obtain by Lemma~\ref{lem33} that
\begin{equation*}
\paren{B\mid\jscript _\alpha\mid A}(\Gamma )=A\sqbrac{\jscript _\alpha}B(\Omega _A\times\Gamma )
   =\rmtr\sqbrac{\,\jscriptbar (\alpha )B(\Gamma )}I
\end{equation*}
(v)\enspace For all $\Delta\in\fscript _\jscript$, $\rho\in\sscript (H)$ we have
\begin{equation*}
\paren{\jscript\mid\hscript _{(\alpha ,A)}}(\Delta )(\rho )=\jscript (\Delta )\sqbrac{\,\hscriptbar _{(\alpha ,A)}(\rho )}=\jscript (\Delta )(\alpha )
   =\jscript _\alpha (\Delta )(\rho )
\end{equation*}
and the result follows.
\end{proof}

An instrument $\iscript$ is \textit{repeatable} if $\rmtr\sqbrac{\iscript (\Delta )\paren{\iscript (\Delta )\rho}}=\rmtr\sqbrac{\iscript (\Delta )(\rho )}$ for all $\Delta\in\fscript _\iscript$, $\rho\in\sscript (H)$ \cite{hz12}.

\begin{thm}    
\label{thm37}
The following statements are equivalent.
{\rm{(i)}}\enspace $\iscript$ is repeatable.
{\rm{(ii)}}\enspace $\iscripthat (\Delta )=(\iscript\circ\iscript )^\wedge (\Delta\times\Delta )$ for all $\Delta\in\fscript _\iscript$.
{\rm{(iii)}}\enspace $(\iscript\circ\iscript )^\wedge (\Delta _1\times\Delta _2)=0$ whenever $\Delta _1\cap\Delta _2=\emptyset$.
{\rm{(iv)}}\enspace $\iscript\circ\iscript (\Delta _1\times\Delta _2)=0$ whenever $\Delta _1\cap\Delta _2=\emptyset$.
{\rm{(v)}}\enspace $(\iscript\circ\iscript )^\wedge (\Delta _1\times\Delta _2)=\iscripthat (\Delta _1\cap\Delta _2)$ for all
$\Delta _1,\Delta _2\in\fscript _\iscript$.
{\rm{(vi)}}\enspace $\iscripthat\sqbrac{\iscript}\iscripthat (\Delta _1\times\Delta _2)=\iscripthat (\Delta _1\cap\Delta _2)$ for all
$\Delta _1,\Delta _2\in\fscript$.
\end{thm}
\begin{proof}
(i)$\Leftrightarrow$(ii) If $\iscript$ is repeatable, then for all $\Delta\in\fscript _\iscript$, $\rho\in\sscript (H)$ we obtain by
Lemma~\ref{lem32} that
\begin{align*}
\rmtr\sqbrac{\rho\iscripthat (\Delta )}&=\rmtr\sqbrac{\iscript (\Delta )(\rho )}=\rmtr\sqbrac{\iscript (\Delta )\paren{\iscript (\Delta )\rho )}}
   =\rmtr\sqbrac{\iscript (\Delta )(\rho )\iscripthat (\Delta )}\\
   &=\rmtr\sqbrac{\rho\iscript ^*(\Delta )\paren{\,\iscripthat (\Delta )}}=\rmtr\sqbrac{\rho (\iscript\circ\iscript )^\wedge (\Delta\times\Delta )}
\end{align*}
Hence, $\iscripthat (\Delta )=(\iscript\circ\iscript )^\wedge (\Delta\times\Delta )$ for all $\Delta\in\fscript _\iscript$.
This also implies the converse.\newline
(iii)$\Leftrightarrow$(ii) Suppose (ii) holds and $\Delta _1\cap\Delta _2=\emptyset$. Then
\begin{align*}
\iscripthat (\Delta _1)+\iscripthat (\Delta _2)&=\iscripthat (\Delta _1\cup\Delta _2)
   =(\iscript\circ\iscript )^\wedge (\Delta _1\cup\Delta _2\times\Delta _1\cup\Delta _2)\\
   &=(\iscript\circ\iscript )^\wedge
   (\Delta _1\times\Delta _1\cup\Delta _2\times\Delta _2\cup\Delta _1\times\Delta _2\cup\Delta _2\times\Delta _1)\\
   &=(\iscript\circ\iscript )^\wedge (\Delta _1\times\Delta _1)+(\iscript\circ\iscript )^\wedge (\Delta _2\times\Delta _2)\\
   &\quad +(\iscript\circ\iscript )^\wedge (\Delta _1\times\Delta _2)+(\iscript\circ\iscript )^\wedge (\Delta _2\times\Delta _1)\\
   &=\iscripthat (\Delta _1)+\iscripthat (\Delta _2)+(\iscript\circ\iscript )^\wedge (\Delta _1\times\Delta _2)
   +(\iscript\circ\iscript )^\wedge (\Delta _2\times\Delta _1)
\end{align*}
It follows that $(\iscript\circ\iscript )^\wedge (\Delta _1\times\Delta _2)=0$. To show the converse, suppose (iii) holds. Denoting the complement of $\Delta$ by $\Delta '$, we obtain
\begin{align*}
\rmtr\sqbrac{\iscript (\Delta )(\rho )}&=\rmtr\sqbrac{\iscript (\Delta\cup\Delta ')(\iscript (\Delta )\rho)}
   =\rmtr\sqbrac{\paren{\iscript (\Delta )+\iscript (\Delta ')}\paren{\iscript (\Delta )\rho}}\\
   &=\rmtr\sqbrac{\iscript (\Delta )\paren{\iscript (\Delta )\rho}}+\rmtr\sqbrac{\iscript (\Delta ')\paren{\iscript (\Delta )\rho}}
\end{align*}
Applying Lemma~\ref{lem32} gives
\begin{align*}
\rmtr\sqbrac{\iscript (\Delta ')\paren{\iscript (\Delta )\rho}}&=\rmtr\sqbrac{\iscript (\Delta )(\rho )\iscripthat (\Delta ')}
   =\rmtr\sqbrac{\rho\iscript ^*(\Delta )\paren{\,\iscripthat (\Delta ')}}\\
   &=\rmtr\sqbrac{\rho(\iscript\circ\iscript )^\wedge (\Delta\times\Delta ')}=0
\end{align*}
Hence, $\rmtr\sqbrac{\iscript (\Delta )(\rho )}=\rmtr\sqbrac{\iscript (\Delta )\paren{\iscript (\Delta )\rho}}$ so (i) and (ii) hold.\newline
(iii)$\Leftrightarrow$(iv) If $\jscript (\Gamma )=0$ then
\begin{equation*}
\rmtr\sqbrac{\rho\jscripthat (\Gamma )}=\rmtr\sqbrac{\jscript (\Gamma )(\rho )}=0
\end{equation*}
so $\jscripthat (\Gamma )=0$. Conversely, if $\jscripthat (\Gamma )=0$, then $\rmtr\sqbrac{\jscript (\Gamma )(\rho )}=0$ for all
$\rho\in\sscript (H)$. Since $\jscript (\Gamma )(\rho )$ is positive, it follows that $\jscript (\Gamma )(\rho )=0$ for all $\rho$ so
$\jscript (\Gamma )=0$. Replacing $\jscript$ with $\iscript\circ\iscript$ gives the result.\newline
(iii)$\Leftrightarrow$(v) Suppose (iii) holds. Since
\begin{align*}
\Delta _1\times\Delta _2&=\Delta _1\times\sqbrac{(\Delta _2\cap\Delta _1)\cup (\Delta _2\cap\Delta '_1)}\\
   &=\sqbrac{\Delta _1\times (\Delta _2\cap\Delta _1)}\cup\sqbrac{\Delta _1\times (\Delta _2\cap\Delta '_1)}\\
\intertext{and}
\sqbrac{\Delta _1\times (\Delta _2\cap\Delta _1)}&\cup\sqbrac{\Delta _1\times (\Delta _2\cap\Delta '_1)}
   =\Delta _1\cap (\Delta _2\cap\Delta '_1)=\emptyset
\end{align*}
we have by (iii) that
\begin{equation*}
(\iscript\circ\iscript )^\wedge\sqbrac{\Delta _1\times (\Delta _2\cap\Delta '_1)}=0
\end{equation*}
Since (iii)$\Rightarrow$(ii) we obtain
\begin{align*}
(\iscript\circ\iscript )^\wedge (\Delta _1\times\Delta _2)&=(\iscript\circ\iscript )^\wedge\sqbrac{\Delta _1\times (\Delta _2\cap\Delta _1)}\\
   &=(\iscript\circ\iscript )^\wedge\sqbrac{(\Delta _1\cap\Delta _2)\cup (\Delta _1\cap\Delta '_2)\times (\Delta _2\cap\Delta _1)}\\
   &=(\iscript\circ\iscript )^\wedge
   \sqbrac{(\Delta _1\cap\Delta _2)\times (\Delta _1\cap\Delta _2)\cup (\Delta _1\cap\Delta '_2)\times (\Delta _1\cap\Delta _2)}\\
   &=(\iscript\circ\iscript )^\wedge\sqbrac{(\Delta _1\cap\Delta _2)\times (\Delta _1\cap\Delta _2)}=\iscripthat (\Delta _1\cap\Delta _2)
\end{align*}
Clearly, (v) implies (iii). (v)$\Leftrightarrow$(vi) This follows because by Lemma~\ref{lem32} we have that
$\iscripthat\sqbrac{\iscript}\iscripthat =(\iscript\circ\iscript )^\wedge$. Reversing the implication shows that (vi) implies (i). Alternatively, since
$\iscripthat\sqbrac{\iscript}\iscripthat =(\iscript\circ\iscript )^\wedge$, letting $\Delta _1=\Delta _2=\Delta$ we obtain from (v) that
\begin{equation*}
\iscripthat (\Delta )=\iscripthat\sqbrac{\iscript}\iscripthat (\Delta\times\Delta )=(\iscript\circ\iscript )^\wedge (\Delta\times\Delta )\qedhere
\end{equation*}
\end{proof}

\begin{cor}    
\label{cor38}
The following statements are equivalent.
{\rm{(i)}}\enspace $\iscript$ is repeatable.
{\rm{(ii)}}\enspace $\iscript ^*(\Delta )I=\iscript ^*(\Delta )\sqbrac{\iscript ^*(\Delta )I}$ for al $\Delta\in\fscript _\iscript$.
{\rm{(iii)}}\enspace $\iscript ^*(\Delta _1)\sqbrac{\iscript ^*(\Delta _2)I}=0$ whenever $\Delta _1\cap\Delta _2=\emptyset$.
{\rm{(iv)}}\enspace $\iscript ^*(\Delta _1)\sqbrac{\iscript ^*(\Delta _2)I}=\iscript ^*(\Delta _1\cap\Delta _2)I$ for all
$\Delta _1,\Delta _2\in\fscript _\iscript$.
\end{cor}
\begin{proof}
By Theorem~\ref{thm37}(ii), $\iscript$ is repeatable if and only if for all $\Delta\in\fscript _\iscript$ we have
\begin{equation*}
\iscript ^*(\Delta )I=\iscripthat (\Delta )=(\iscript\circ\iscript )^*(\Delta\times\Delta )=\iscript ^*(\Delta )\sqbrac{\,\iscripthat (\Delta )}
   =\iscript ^*\sqbrac{\iscript ^*(\Delta )I}
\end{equation*}
By Theorem~\ref{thm37}(iii), $\iscript$ is repeatable if and only if whenever $\Delta _1\cap\Delta _2=\emptyset$ we have
\begin{equation*}
\iscript ^*(\Delta _1)\sqbrac{\iscript ^*(\Delta _2)I}=(\iscript\circ\iscript )^\wedge (\Delta _1\times\Delta _2)=0
\end{equation*}
By Theorem~\ref{thm37}(v), $\iscript$ is repeatable if and only if for all $\Delta _1,\Delta _2\in\fscript _\iscript$ we have
\begin{align*}
\iscript ^*(\Delta _1)\sqbrac{\iscript ^*(\Delta _2)I}&=(\iscript\circ\iscript )^\wedge (\Delta _1\times\Delta _2)
   =\iscripthat (\Delta _1\cap\Delta _2)\\
   &=\iscript ^*(\Delta _1\cap\Delta _2)I\qedhere
\end{align*}
\end{proof}

\section{Finite Instruments and Observables}  
We now consider finite instruments and observables. One of the main advantages of the finite case is that we can introduce L\"uders instruments \cite{lud51} which do not seem to exist in the infinite case. Although finiteness is a strong assumption, it is general enough to include quantum computation and information theory \cite{hz12,kw17,nc00}. For a finite set $\Omega =\brac{x_1,x_2,\ldots ,x_n}$ we assume that the corresponding $\sigma$-algebra is $2^\Omega$ so the outcome space is specified by $\Omega$ and we need not mention the
$\sigma$-algebra. A \textit{finite instrument} with outcome space $\Omega$ corresponds to a set
\begin{equation*}
\iscript =\brac{\iscript _{x_1},\iscript _{x_2},\ldots ,\iscript _{x_n}}\subseteq\oscript (H)
\end{equation*}
for which $\iscriptbar =\sum\limits _{i=1}^n\iscript _{x_i}$ is a channel. We then define
$\iscript (\Delta )=\sum\limits _{x_i\in\Delta}\iscript _{x_i}$ for all $\Delta\subseteq\Omega$ so $\Delta\mapsto\iscript (\Delta )$ becomes an instrument \cite{gud120,gud220,hz12,nc00}. Similarly, a \textit{finite observable} with outcome space $\Omega$ corresponds to a set
$A=\brac{A_{x_1},A_{x_2},\ldots A_{x_n}}\subseteq\escript (H)$ for which $\sum\limits _{i=1}^nA_{x_i}=I$. We again define
$A(\Delta )=\sum\limits _{x_i\in\Delta}A_{x_i}$ and $\Delta\mapsto A(\Delta )$ becomes an observable. As before, an instrument $\iscript$ measures a unique observable $\iscripthat$ that satisfies $\rmtr (\rho\iscripthat _{x_i})=\rmtr\sqbrac{\iscript _{x_i}(\rho )}$, $i=1,2,\ldots ,n$,
$\rho\in\sscript (H)$. Of course, this is equivalent to
\begin{equation*}
\rmtr\sqbrac{\rho\iscripthat (\Delta )}=\rmtr\sqbrac{\iscript (\Delta )(\rho )}
\end{equation*}
for all $\Delta\subseteq\Omega$, $\rho\in\sscript (H)$. For conciseness, we use the notion
\begin{equation*}
\iscript (x)=\iscript\paren{\brac{x}}=\iscript _x
\end{equation*}

\begin{thm}    
\label{thm41}
A finite instrument $\iscript$ is repeatable if and only if
\begin{equation*}
\rmtr\sqbrac{\iscript _x(\rho )}=\rmtr\sqbrac{\iscript _x\paren{\iscript _x(\rho )}}
\end{equation*}
for all $x\in\Omega _\iscript$, $\rho\in\sscript (H)$.
\end{thm}
\begin{proof}
If $\iscript$ is repeatable, then
\begin{equation*}
\rmtr\sqbrac{\iscript _x(\rho )}=\rmtr\sqbrac{\iscript (x)(\rho )}=\rmtr\sqbrac{\iscript (x)\paren{\iscript (x)(\rho )}}
   =\rmtr\sqbrac{\iscript _x\paren{\iscript _x(\rho )}}
\end{equation*}
for all $x\in\Omega _\iscript$, $\rho\in\sscript (H)$. Conversely, suppose
$\rmtr\sqbrac{\iscript _x(\rho )}=\rmtr\sqbrac{\iscript _x\paren{\iscript _x(\rho )}}$ holds. Since
\begin{equation*}
\sum _y\rmtr\sqbrac{\iscript _y\paren{\iscript _x(\rho )}}=\rmtr\sqbrac{\,\iscriptbar\paren{\iscript _x(\rho )}}
   =\rmtr\sqbrac{\iscript _x(\rho )}
\end{equation*}
we conclude that $\sum\limits _{y\ne x}\rmtr\sqbrac{\iscript _y\paren{\iscript _x(\rho )}}=0$ so
$\rmtr\sqbrac{\iscript _y\paren{\iscript _x(\rho )}}=0$ for all $\rho\in\sscript (H)$ and $y\ne x$. We conclude that
\begin{align*}
\rmtr\sqbrac{\iscript (\Delta )\paren{\iscript (\Delta )(\rho )}}
  &=\rmtr\sqbrac{\paren{\sum _{y\in\Delta}\iscript _y}\paren{\sum _{x\in\Delta}\iscript _x(\rho )}}
  =\sum _{x,y\in\Delta}\rmtr\sqbrac{\iscript _y\paren{\iscript _x(\rho )}}\\
  &=\sum _{x\in\Delta}\rmtr\sqbrac{\iscript _x\paren{\iscript _x(\rho )}}=\sum _{x\in\Delta}\rmtr\sqbrac{\iscript _x(\rho )}
  =\rmtr\sqbrac{\iscript (\Delta )(\rho )}
\end{align*}
for all $\Delta\subseteq\Omega _\iscript$, $\rho\in\sscript (H)$. Hence, $\iscript$ is repeatable.
\end{proof}

The instrument $\iscript\circ\jscript$ and observables $A\sqbrac{\iscript}B$ are determined by their outcomes
$(\iscript\circ\jscript )_{(x,y)}=\iscript _x\circ\jscript _y$ and $\paren{A\sqbrac{\iscript}B}_{(x,y)}=\iscript _x^*(B_y)$. The next result follows from Theorem~\ref{thm37}

\begin{cor}    
\label{cor42}
The following statements for a finite instrument $\iscript$ are equivalent.
{\rm{(i)}}\enspace $\iscript$ is repeatable.
{\rm{(ii)}}\enspace $\iscripthat _x=(\iscript\circ\iscript )_{(x,x)}^\wedge$ for all $x\in\Omega _\iscript$.
{\rm{(iii)}}\enspace $(\iscript\circ\iscript )_{(x,y)}^\wedge =0$ if $x\ne y$.
{\rm{(iv)}}\enspace $(\iscript\circ\iscript )_{(x,y)}^\wedge =\iscripthat\paren{\brac{x}\cap\brac{y}}$ for all $x,y\in\Omega _\iscript$.
{\rm{(v)}}\enspace For all $x,y\in\Omega _\iscript$ we have
\begin{equation*}
\paren{\,\iscripthat\sqbrac{\iscript}\,\iscripthat\,}_{(x,y)}=\iscripthat\paren{\brac{x}\cap\brac{y}})
\end{equation*}
\end{cor}

We now consider a generalization of a Holevo instrument for the finite case. If $A=\brac{A_x\colon x\in\Omega}$ is a finite observable and
$\alpha _x\in\sscript (H)$, $x\in\Omega$, then the instrument
\begin{equation*}
\sqbrac{\hscript _{(\alpha ,A)}}_x(\rho )=\rmtr (\rho A_x)\alpha _x
\end{equation*}
is called a (\textit{finite}) \textit{Holevo instrument with states} $\alpha _x$ \textit{and observable} $A$. The instrument $\hscript _{(\alpha ,A)}$ is also called a \textit{conditional state preparator} \cite{hz12}.

\begin{lem}    
\label{lem43}
A Holevo instrument $\hscript _{(\alpha ,A)}$ is repeatable if and only if\ $\rmtr (\alpha _xA_x)=1$ for all $x$ with $A_x\ne 0$.
\end{lem}
\begin{proof}
For all $\rho\in\sscript (H)$, $x\in\Omega$, writing $\iscript =\hscript _{(\alpha ,A)}$ we obtain
\begin{equation*}
\rmtr\sqbrac{\iscript _x\paren{\iscript _x(\rho )}}=\rmtr\sqbrac{\iscript _x\paren{\rmtr (\rho A_x)\alpha _x}}
   =\rmtr (\rho A_x)\rmtr\sqbrac{\iscript _x(\alpha _x)}=\rmtr (\rho A_x)\rmtr (\alpha _xA_x)
\end{equation*}
Hence, $\iscript$ is repeatable if and only if
\begin{equation*}
\rmtr\sqbrac{\iscript _x\paren{\iscript _x(\rho )}}=\rmtr\sqbrac{\iscript _x(\rho )}=\rmtr (\rho A_x)
\end{equation*}
which is equivalent to $\rmtr (\rho A_x)\rmtr (\alpha _xA_x)=\rmtr (\rho A_x)$ for all $\rho\in\sscript (H)$, $x\in\Omega$. Choosing $\rho$ such that $\rmtr (\rho A_x)\ne 0$ we conclude that $\rmtr (\alpha _xA_x)=1$ for all $x$ satisfying $A_x\ne 0$.
\end{proof}

In Lemma~\ref{lem43} we can choose $\alpha _x=\ket{\psi _x}\bra{\psi _x}$ where $\ket{\psi _x}$ is a unit eigenvector for $A_x$. We now generalize Theorem~\ref{thm35} for finite Holevo instruments.

\begin{thm}    
\label{thm44}
{\rm{(i)}}\enspace $(\hscript _{(\alpha ,A)}^*)_x(a)=\rmtr (\alpha _xa)A_x$ and $\hscripthat _{(\alpha ,A)}=A$ so $\hscript _{(\alpha ,A)}$ measures $A$.
{\rm{(ii)}}\enspace If $\iscript\in\inset (H)$ is finite, then $\iscript\circ\hscript _{(\alpha ,A)}$ is a Holevo instrument with states $\alpha _y$ and observable $B_{(x,y)}=\iscript _x^*(A_y)$.
{\rm{(iii)}}\enspace If $\iscript\in\inset (H)$ is finite, then $\hscript _{(\alpha ,A)}\circ\iscript$ is a Holevo instrument with states
$\iscript _y(\alpha _x)^\sim$ where $\rmtr\sqbrac{\iscript _y(\alpha _x)}\ne 0$ and observable
$B_{(x,y)}=\rmtr\sqbrac{\iscript _y(\alpha _x)}A_x$.
{\rm{(iv)}}\enspace $\hscript _{(\beta ,B)}\circ\hscript _{(\alpha ,A)}$ is a finite Holevo instrument with states $\alpha _y$ and observable
$C _{(x,y)}=\rmtr (\beta _xA_y)B_x$.
{\rm{(v)}}\enspace $\paren{B\mid\hscript _{(\alpha ,A)}\mid A}_y=\sum\limits _x\rmtr (\alpha _xB_y)A_x$.
\end{thm}
\begin{proof}
(i)\enspace For every $\rho\in\sscript (H)$ $x\in\Omega _A$, $a\in\escript (H)$ we have
\begin{align*}
\rmtr\sqbrac{\rho (\hscript _{(\alpha ,A)}^*)_x(a)}&=\rmtr\sqbrac{(\hscript _{(\alpha ,A)})_x(\rho )a}=\rmtr\sqbrac{\rmtr (\rho A_x)\alpha _xa}\\
   &=\rmtr (\rho A_x)\rmtr (\alpha _xa)=\rmtr\sqbrac{\rho\rmtr (\alpha _xa)A_x}
\end{align*}
Hence, $(\hscript _{(\alpha ,A)}^*)_x(a)=\rmtr (\alpha _xa)A_x$. Moreover,
\begin{equation*}
(\,\hscripthat _{(\alpha ,A)})_x=\rmtr (\hscript ^*_{(\alpha ,A)})_x(I)=A_x
\end{equation*}
so $\hscripthat _{(\alpha ,A)}=A$.
(ii)\enspace For all $x\in\Omega _\iscript$, $y\in\Omega _A$, $\rho\in\sscript (H)$ we obtain
\begin{align*}
(\iscript\circ\hscript _{(\alpha ,A)})_{(x,y)}(\rho )&=\iscript _x\circ (\hscript _{(\alpha ,A)})_y(\rho )
   =(\hscript _{(\alpha ,A)})_y\paren{\iscript _x(\rho )}\\
   &=\rmtr\sqbrac{\iscript _x(\rho )A_y}\alpha _y=\rmtr\sqbrac{\rho\paren{\iscript _x^*(A_y)}}\alpha _y
\end{align*}
Notice that $B_{(x,y)}=\iscript _x^*(A_y)$ is an observable because $\iscript _x^*(A_y)\in\escript (H)$ and
\begin{equation*}
\sum _{x,y}B_{(x,y)}=\sum _{x,y}\iscript _x^*(A_y)=\sum _x\iscript _x^*\paren{\sum _yA_y}=\sum _x\iscript\,^*_x(I)
   =\iscript (\Omega )^*(I)=I
\end{equation*}
(iii)\enspace For all $x\in\Omega _A$, $y\in\Omega _\iscript$, $\rho\in\sscript (H)$ we obtain
\begin{align*}
(\hscript _{(\alpha ,A)}\circ\iscript )_{(x,y)}(\rho )&=(\hscript _{(\alpha ,A)})_x\circ\iscript _y(\rho )
   =\iscript _y\sqbrac{(\hscript _{(\alpha ,A)})_x(\rho )}\\
   &=\rmtr (\rho A_x)\iscript _y(\alpha _x)=\rmtr\sqbrac{\rho\rmtr\paren{\iscript _y(\alpha _x)}A_x}\iscript _y(\alpha _x)^\sim
\end{align*}
Notice that $B_{(x,y)}=\rmtr\sqbrac{\iscript _y(\alpha _x)}A_x$ is an observable because
$\rmtr\sqbrac{\iscript _y(\alpha _x)}A_x\in\escript (H)$ and
\begin{align*}
\sum _{x,y}B_{(x,y)}&=\sum _{x,y}\rmtr\sqbrac{\iscript _y(\alpha _x)}A_x=\sum _x\rmtr\sqbrac{\sum _y\iscript _y(x)}A_x
   =\sum _x\rmtr\sqbrac{\,\iscriptbar (\alpha _x)}A_x\\
   &=\sum _xA_x=I
\end{align*}
(iv)\enspace Applying (ii) we obtain for all $x\in\Omega _B$, $y\in\Omega _A$, $\rho\in\sscript (H)$ that
\begin{align*}
(\hscript _{(\beta ,B)}\circ\hscript _{(\alpha ,A)})_{(x,y)}(\rho )&=\rmtr\sqbrac{\rho (\hscript _{(\beta ,B)}^*)_xA_y}\alpha _y
   =\rmtr\sqbrac{\rho\rmtr (\beta _xA_y)B_x}\alpha _y\\
   &=\rmtr (\beta _xA_y)\rmtr (\rho B_x)\alpha _y=\rmtr\sqbrac{\rho\rmtr (\beta _xA_y)B_x}\alpha _y
\end{align*}
Notice that $C_{(x,y)}=\rmtr (\beta _xA_y)B_x$ is an observable because $\rmtr (\beta _xA_y)B_x\in\escript (H)$ and
\begin{align*}
\sum _{x,y}C_{(x,y)}&=\sum _{x,y}\rmtr (\beta _xA_y)B_x=\sum _x\rmtr\paren{\beta _x\sum _yA_y}B_x
   =\sum _x\rmtr (\beta _xI)B_x\\
   &=\sum _xB_x=I
\end{align*}
(v)\enspace For all $y\in\Omega _B$, applying (i) we obtain
\begin{equation*}
\paren{B\mid\hscript _{(\alpha ,A)}\mid A}_y=\hscriptbar _{(\alpha ,A)}^*(B_y)=\sum _x\rmtr (\alpha _xB_y)A_x\qedhere
\end{equation*}
\end{proof}

If $A=\brac{A_x\colon x\in\Omega _A}$ is a finite observable, we define the corresponding \textit{L\"uders instrument} \cite{kra83,lud51} with outcome space $\Omega _A$ by
\begin{equation*}
\lscript _x^A(\rho )=A_x^{1/2}\rho A_x^{1/2}
\end{equation*}
for all $x\in\Omega _A$. We then have for all $\Delta\subseteq\Omega _A$ that
\begin{equation*}
\lscript ^A(\Delta )=\sum\brac{A_x^{1/2}\rho A_x^{1/2}\colon x\in\Delta}
\end{equation*}
We now generalize Theorem~\ref{thm21} to instruments.

\begin{thm}    
\label{thm45}
Let $A,B\in\obset (H)$, $\jscript\in\inset (H)$ be finite.
{\rm{(i)}}\enspace $(\lscript ^A)_x^*(a)=A_x^{1/2}aA_x^{1/2}$ and $(\lscript ^A)^\wedge =A$ so $\lscript ^A$ measures $A$.
{\rm{(ii)}}\enspace $(\lscript ^A\circ\jscript )_{(x,y)}^\wedge =A_x^{1/2}\jscripthat _yA_x^{1/2}$ for all $x\in\Omega _A$, $y\in\Omega _\jscript$.
{\rm{(iii)}}\enspace $(\jscript\circ\lscript ^A)_{(y,x)}^\wedge =\jscript _y^*(A_x)$ for all $x\in\Omega _X$, $y\in\Omega _\jscript$.
{\rm{(iv)}}\enspace $(\jscript\mid\lscript ^A)_y^\wedge =\sum\limits _{x\in\Omega _A}A_x^{1/2}\jscripthat _yA_x^{1/2}$.
{\rm{(v)}}\enspace $(\lscript ^A\mid\jscript )_x^\wedge =\jscriptbar\,^*(A_x)$.
{\rm{(vi)}}\enspace $\paren{A\sqbrac{\lscript ^A}B}_{(x,y)}=A_x^{1/2}B_yA_x^{1/2}$.
{\rm{(vii)}}\enspace $\paren{B\mid\lscript ^A\mid A}_y=\sum\limits _{x\in\Omega _A}A_x^{1/2}B_yA_x^{1/2}$.
\end{thm}
\begin{proof}
(i)\enspace For all $x\in\Omega _A$, $a\in\escript (H)$, $\rho\in\sscript (H)$ we have
\begin{equation*}
\rmtr\sqbrac{\rho (\lscript ^A)_x^*(a)}=\rmtr\sqbrac{\lscript _x^A(\rho )a}=\rmtr (A_x^{1/2}\rho A_x^{1/2}a)
=\rmtr\sqbrac{\rho A_x^{1/2}aA_x^{1/2}}
\end{equation*}
Hence, $(\lscript ^A)_x^*(a)=A_x^{1/2}aA_x^{1/2}$. Moreover,
\begin{equation*}
(\lscript ^A)_x^\wedge =(\lscript _x^A)^*(I)=A_x^{1/2}IA_x^{1/2}=A_x
\end{equation*}
Therefore, $(\lscript ^A)^\wedge =A$ so $\lscript ^A$ measures $A$.
(ii)\enspace By Theorem~\ref{thm21}(iii) we obtain
\begin{equation*}
(\lscript ^A\circ\jscript )_{(x,y)}^\wedge =(\lscript ^A)_x^*(\,\jscripthat _y)=A^{1/2}\jscripthat _yA_x^{1/2}
\end{equation*}
(iii)\enspace Applying Theorem~\ref{thm21} we have the following
\begin{equation*}
(\jscript\circ\lscript ^A)_{(y,x)}^\wedge =\jscript _y^*(\,\lscripthat _x^A)=\jscript _y^*(A_x)
\end{equation*}
(iv)\enspace $(\jscript\mid\lscript ^A)_y^A=(\,\lscriptbar ^A)^*(\,\jscripthat _y)=\sum\limits _{x\in\Omega _A}A_x^{1/2}\jscripthat _yA_x^{1/2}$
\newline
(v)\enspace $(\lscript ^A\mid\jscript )_x^\wedge =\jscripthat ^*(\lscript ^A)_x^\wedge
=\sum\limits _{y\in\Omega _\jscript}\jscript _y^*(A_x)$.\newline
(vi)\enspace $\paren{A\sqbrac{\lscript ^A}B}_{(x,y)}=(\lscript _x^A)^*(B_y)=A_x^{1/2}B_yA_x^{1/2}$.\newline
(vii)\enspace $(B\mid\lscript ^A\mid A)_y=\paren{A\sqbrac{\lscript ^A}B}(\Omega _A\times\brac{y})
=\sum\limits _{x\in\Omega _A}A_x^{1/2}B_yA_x^{1/2}$.
\end{proof}

\begin{cor}    
\label{cor46}
Let $A,B\in\oscript (H)$, $\jscript\in\inset (H)$ be finite and let $\Delta\subseteq\Omega _A$, $\Gamma\subseteq\Omega _B$.
{\rm{(i)}}\enspace $(\lscript ^A)^*(\Delta )(a)=\sum\limits _{x\in\Delta}(A_x\square a)$.
{\rm{(ii)}}\enspace $(\lscript ^A\circ\jscript )^\wedge (\Delta\times\Gamma )=\sum\limits _{x\in\Delta}\sqbrac{A_x\square\jscripthat (\Gamma )}$.
{\rm{(iii)}}\enspace $(\jscript\circ\lscript ^A)^\wedge (\Gamma\times\Delta )=\jscript ^*(\Gamma )\sqbrac{A(\Delta )}$.\newline
{\rm{(iv)}}\enspace $(\jscript\mid\lscript ^A)^\wedge (\Gamma )=\sum\limits _{x\in\Omega _A}(A_x\square\jscripthat (\Gamma ))$
{\rm{(v)}}\enspace $(\lscript ^A\mid\jscript )^\wedge (\Delta )=\jscriptbar\,^*\sqbrac{A(\Delta )}$.
{\rm{(vi)}}\enspace $\paren{A\sqbrac{\lscript ^A}B}(\Delta\times\Gamma )=\sum\limits _{x\in\Delta}\sqbrac{A_x\square B(\Gamma )}$.
\newline
{\rm{(vii)}}\enspace $(B\mid\lscript ^A\mid A)(\Gamma )=\sum\limits _{x\in\Omega _A}\sqbrac{A_x\square B(\Gamma )}$.
\end{cor}

We close by stating that a L\"uders instrument is repeatable if and only if it is sharp \cite{hz12}.

\end{document}